\documentclass[11pt]{article}
\usepackage{amsmath,amssymb,amsthm,latexsym,color}
\renewcommand{\baselinestretch}{1.4}
\usepackage{graphicx}
\usepackage{color,latexsym,amsfonts}
\usepackage{epsfig,euscript}

\renewcommand{\hat}{\widehat}
\def\singlespace{\def\baselinestretch{1}\@normalsize}

\def\wh{\widehat}
\def\wt{\widetilde}

\newcommand{\cov}{{\rm Cov}}

\newcommand{\tr}{\mbox{tr}}
\newcommand{\var}{\mbox{Var}}

\def\ga{\gamma}

\newcommand{\bA}{{\mathbf A}}
\newcommand{\bB}{{\mathbf B}}
\newcommand{\bF}{{\mathbf F}}

\newcommand{\bH}{{\mathbf H}}
\newcommand{\bI}{{\mathbf I}}
\newcommand{\bK}{{\mathbf K}}
\newcommand{\bL}{{\mathbf L}}

\newcommand{\bQ}{{\mathbf Q}}
\newcommand{\bP}{{\mathbf P}}

\newcommand{\bS}{{\mathbf S}}
\newcommand{\bU}{{\mathbf U}}

\newcommand{\bW}{{\mathbf W}}

\newcommand{\bZ}{{\mathbf Z}}
\newcommand{\ba}{{\mathbf a}}

\newcommand{\bt}{{\mathbf t}}

\newcommand{\bu}{{\mathbf u}}
\newcommand{\bv}{{\mathbf v}}

\newcommand{\bx}{{\mathbf x}}
\newcommand{\by}{{\mathbf y}}
\newcommand{\bz}{{\mathbf z}}

\newcommand{\bfeta}  {\boldsymbol{\eta}}

\newcommand{\bOmega}{\boldsymbol{\Omega}}

\newcommand{\bSigma}{\boldsymbol{\Sigma}}

\newcommand{\bgamma}{\boldsymbol{\gamma}}

\newcommand{\bve}{\mbox{\boldmath$\varepsilon$}}

\newcommand{\bTheta} {\boldsymbol{\Theta}}

\newcommand{\bxi} {\boldsymbol{\xi}}
\newcommand{\bmu} {\boldsymbol{\mu}}

\newcommand{\bGamma} {\boldsymbol{\Gamma}}
\newcommand{\bLambda} {\boldsymbol{\Lambda}}
\newcommand{\bC}{{\mathbf C}}
\newcommand{\bD}{{\mathbf D}}

\newcommand{\calM}{{\mathcal M}}

\def\6bullets{\bullet\bullet\bullet\bullet\bullet\bullet}

\def\AS{{\sl The Annals of Statistics}}

\def\SPA{{\sl Stochastic Processes and Their Applications}}

\def\JSPI{{\sl Journal of Statistical Planning and Inference}}

\def\JOE{{\sl Journal of Econometrics }}

\newcommand{\Date}[1]{\def\@Date{#1}}
\def\today{\number\day~\ifcase\month\or
 January\or February\or March\or April\or May\or June\or
 July\or August\or September\or October\or November\or December\fi~\number\year}

\newtheorem{theorem}{Theorem}
\newtheorem{lemma}[theorem]{Lemma}

\newtheorem{remark}{Remark} 

\newcommand{\nn}{\nonumber}

\newcommand{\be}{\begin{array}}
\newcommand{\ee}{\end{array}}

\newcommand{\ignore}[1]{}{}

\newcommand{\vp}{\varepsilon}
\newcommand{\boldX}{{\mbox{\boldmath $X$}}}

\newcommand{\boldeta}{{\mbox{\boldmath $\eta$}}}

\newcommand{\boldmu}{\mbox{\boldmath $\mu$}}
\newcommand{\bolde}{\mbox{\boldmath $\varepsilon$}}

\newcommand{\boldx}{\mbox{\boldmath $x$}}

\newcommand{\boldZ}{\mbox{\boldmath $Z$}}

\newcommand{\boldxi}{\mbox{\boldmath $\xi$}}
\newcommand{\bfe}{\mbox{\boldmath $e$}}
\newcommand{\idv}{\mathbf{1}}
\renewcommand{\tilde}{\widetilde}

    {\setcounter{araenumcount}{0}%
    \begin{list}{\arabic{araenumcount}.}{\usecounter{araenumcount}}}%
    {\end{list}}
    {\setcounter{roenumcount}{0}%
    \begin{list}{(\Roman{roenumcount})}{\usecounter{roenumcount}}}%
    {\end{list}}
    {\setcounter{romenumcount}{0}%
    \begin{list}{(\roman{romenumcount})}{\usecounter{romenumcount}}}%
    {\end{list}}

\newcommand{\beqn}{\begin{eqnarray}}             
\newcommand{\eeqn}{\end{eqnarray}}               
\newcommand{\beq}{\begin{eqnarray*}}             
\newcommand{\eeq}{\end{eqnarray*}}

\begin{document}

\title{\bf Identifying Cointegration by Eigenanalysis\footnote{
Rongmao Zhang is Professor, School of Mathematics, Zhejiang University,
Hangzhou, China (E-mail: rmzhang@zju.edu.cn).
Peter Robinson is Tooke Professor of Economic Science and Statistics,
Department of Economics, London School of Economics, UK (E-mail: p.m.robinson@lse.ac.uk).
Qiwei Yao is Professor,
Department of Statistics, London School of Economics, Houghton Street, London,
WC2A 2AE, U.K.; Guanghua School of Management, Peking University,
China (E-mail: q.yao@lse.ac.uk). Partially supported by
the NSFC research grants 11371318/11771390, the ZPNSFC research grant LR16A010001, the ESRC research grant ES/J007242/1, and
 the EPSRC research grant EP/L01226X/1.}}

\author{ Rongmao Zhang \qquad Peter Robinson \qquad
Qiwei Yao\\[2ex]
 }
 \date{}

\maketitle

\begin{abstract}
We propose a new and easy-to-use method for identifying cointegrated
components of nonstationary time series, consisting of an eigenanalysis for a
certain non-negative definite matrix.  Our setting is model-free, and we
allow the integer-valued integration orders of the
observable series to be unknown, and to possibly differ.  Consistency of estimates of the
cointegration space and cointegration rank is established both when the
dimension of the observable time series is fixed as sample size increases,
and when it diverges slowly.
The proposed methodology is also extended and justified
in a fractional setting.
A Monte Carlo study of finite-sample performance, and a small empirical
illustration, are reported.
\end{abstract}


\vskip3mm
\begin{quote}
\noindent
{\sc Keywords}:  Cointegration, Eigenanalysis, $I(d)$,
Nonstationary processes,  Vector time series.
\end{quote}

\newpage

\section{Introduction}

Cointegration entails a dimensionality reduction of certain observable
multiple time series that are dominated by common components. In particular
a multiple time series can be said to be (linearly) cointegrated if there
exists an instantaneous linear combination, or cointegrating error, with
lower integration order. Much of the vast literature, following Box and Tiao
(1977), Granger (1981), Engle and Granger (1987), has focused on unit root
series which have one or more short memory cointegrating errors, but there
have been extensions to nonstationary series with other integer orders of
integration, allowing also for the possibility of some nonstationary
cointegrating errors, as well as to fractional nonstationary, and even
stationary, observable series and cointegrating errors, with unknown
integration orders.  Much of the early literature, in particular, assumed a
complete parameterization of second order properties, where in particular
the observable series are generated from short memory inputs that have
finite autoregressive moving average (ARMA) structure, but it has also been
common to study semiparametric settings, with underlying short memory inputs
having nonparametric autocorrelation, see e.g. Stock (1987), Phillips
(1991), in some cases without sacrificing precision relative to a correctly
specified parametric structure.

Given knowledge of the cointegration rank, $r,$ of a $p$-dimensional
observable series, that is the number of cointegrating relations, various
methods are available for estimating the unknown parameters of the model,
such as the coefficients of the cointegrating errors, and even of unknown
integration orders, and for carrying out asymptotically valid, and sometimes
even efficient, statistical inference. However, $r$ might not be known to
the practitioner, and various approaches for estimating $r$ from the data
have been developed, starting from Engle and Granger (1987), Johansen
(1991), in their parametric, unit root vector autoregressive (VAR) setting,
\ and continuing with, for example, Aznar and Salvador (2002) and Saikkonen
and L\"utkepohl (2000).  If, however, the order of the VAR is underspecified,
or all observable series do not have a single unit root, then typically the
resulting specification error will invalidate such approaches, not to
mention rules of statistical inference on unknown coefficients in the model.
It is possible that one or more of the nonstationary observable processes
could have two or more unit roots, or indeed could have fractional orders of
integration, as supported by some empirical investigations.  References
that allow for nonparametric autocorrelation and/or unknown integration
orders include Phillips and Ouliaris (1988, 1990), Bierens (1997), Stock (1999), Shintani
(2001), Harris and Poskitt (2004), Li, Pan and Yao (2009) in the case of
integer integration orders, and Robinson and Yajima (2002),  Chen and
Hurvich (2006),  Robinson (2008) in case of fractional integration
orders, including in the latter setting cases where observables are
stationary and the cointegrating errors are stationary with less memory.

Like  Phillips and Ouliaris (1988), Robinson and Yajima (2002), Harris and
Poskitt (2004), Li, Pan and Yao (2009), we employ methods based on eigenanalysis.
In our case, in the setting of nonparametric
autocorrelation and unknown (and possibly different) integration orders, we
employ eigenvalues of a
certain non-negative definite matrix function of sample autocovariance
matrices of the observable series, 
for estimating
cointegration rank, with the cointegration space then estimated by
selection of eigenvectors, and cointegrating errors thereby proxied.
Though the initial development assumes that observable series have integer
orders and cointegrating errors have short memory, we extend these results
to allow for observables to be fractionally nonstationary, and cointegrating
errors to be fractionally stationary.  In both circumstances we establish
consistency of our estimates of cointegration rank and space with $p$
fixed as the length $n$ of our time series  diverges.  In case of integer
integration orders, we also establish consistency allowing $p$ to diverge
slowly with $n.$

The rest of the paper is organized as follows. The proposed methodology
is presented  in Section 2. Asymptotic theory with integer orders of
integration is developed in Section 3.
 In Section 4, both the proposed method and part of the asymptotic theory
are extended to the fractional case.
Simulations and a small real data  example are reported in Section 5.
All statements and proofs  are relegated to  an Appendix, which also
contains  a number of technical lemmas.

\section{Methods}
\setcounter{equation}{0}

\subsection{Setting}
We call a vector process $\bu_t$ weakly stationary if
(i) $E\bu_t$ is a constant vector independent of $t$, and (ii) $E\|\bu_t\|^2 < \infty$,
and $\cov(\bu_t, \bu_{t+s}) $ depends on $s$ only for any integers $t, s$, where
$\| \cdot \|$ denotes the Euclidean norm.
Denote by $\nabla$ the difference
operator, i.e. $\nabla \bu_t = \bu_t-\bu_{t-1}$, and $\nabla^d \bu_t =
\nabla(\nabla^{d-1} \bu_t)$ for any integer $d\ge 1$. We use the convention
$\nabla^0 \bu_t = \bu_t $. Further, if $\mathbf{u}_{t}$ has spectral density matrix that is finite and
positive definite at zero frequency we say $\mathbf{u}_{t}$ is an $I\left(
0\right) $ process.  An example of an $I\left( 0\right) $ process is a
stationary an invertible vector ARMA, and many $I\left( 0\right) $ processes
satisfy Condition 1 of Section 3.1 below, imposed for our asymptotic theory,
including the examples described immediately after Condition 1.  Now denote
by $u_{it}$ the $i$th element of $\mathbf{u}_{t}$ and define $%
u_{it}^{+}=u_{it}1\left( t\geq 1\right) ,$ where $1\left( \cdot\right) $ is the
indicator function. For an $m$-dimensional $I\left( 0\right) $ process $%
\mathbf{u}_{t}$ and non-negative integers $d_{1},...,d_{m},$ we say that $%
\mathbf{v}_{t}$ $=\left( \nabla ^{-d_{1}}u_{1t}^{+},...,\nabla
^{-d_{m}}u_{mt}^{+}\right) ^{\prime }$ is an ($m$-dimensional) $I\left(
d_{1},...,d_{m}\right) $ process, with some abuse of notation when $m=1,$ $%
d_{1}$ $=0$.  Note that for $d_{1}=...=d_{m}=0,$ $\mathbf{v}_{t}$ is not $%
I\left( 0\right) $ or even weakly stationary or equivalent to $\mathbf{u}_{t}
$ due to the truncation (implying $\mathbf{v}_{t}=0,$ $t\leq 0$) that is
imposed in order to achieve bounded variance in case of positive $d_{i},$
but it is `asymptotically' weakly stationary and $I\left( 0\right) .$
When $d_{1}=...=d_{m}=1,$ all elements of $\mathbf{v}_{t}$ have a single
unit root,  but we are concerned with processes for which $d_{i}$ can vary
over $i.$ \ \

Now assume a $p\times 1$ \ observable time series \ $\mathbf{y}_{t}$ is $%
I\left( d_{1},...,d_{p}\right) $ for non-negative integers, and
 admits the  form
\begin{equation}
\label{2.1}
\bold{y}_t=\bA\bold{x}_t,
\end{equation}
where $\bA$ is an unknown and invertible constant matrix, $\bold{x}_t=(\bold{x}'_{t1},
\bold{x}'_{t2})'$ is a latent $p\times 1$ process,  $\bold{x}_{t2}$ is an
 $r\times 1$  $I(0)$ process, and $\mathbf{x}_{t1}$ \ is an $I\left( c_{1},...,c_{p-r}\right) $ process,
where each $c_{i}$ is an element of the set $\left\{ d_{1},...,d_{p}\right\}
.$
Furthermore  no linear combination of $\bx_{t1}$ is $I(0)$,
as such a stationary variable can be absorbed into $\bx_{t2}$.
Each component of $\bx_{t2}$ is a cointegrating error  of $\by_t$
and $r\ge 0$ is the cointegration rank.
In the event that there exists no cointegration among the components of $\by_t$,
$r=0$. When $\by_t$ itself is $I(0, \cdots, 0)$, $r=p$.
But these are two extreme cases. Note that cointegration requires equality of at least two $d_i.$ For many economic and financial applications,
there exist a small number of cointegrated variables, i.e. $r \ge 1$ is a small
integer.

The pair $(\bA, \bold{x}_t)$
 in (\ref{2.1})
is not uniquely defined,
as it can be replaced by $(\bA\bH^{-1}, \bH\bold{x}_t)$ for any
invertible $\bH$ of the form
\beqn \left(\begin{array}{cc} \bH_{11} &
\bH_{12}\\
\bf{0} & \bH_{22}
\end{array}\right)\nn\eeqn
where $\bH_{11}, \bH_{22}$ are square matrices of size $(p-r), \,r$
respectively, and $\bf{0}$ denotes a matrix with all entries equal to
$0.$
Therefore there is no loss of generality in assuming $\bA$ to be
orthogonal, because any non-orthogonal $\bA$ admits the
decomposition $\bA=\bQ\bU,$ where $\bQ$ is  orthogonal and $\bU$ is
upper-triangular, and we may then replace $(\bA, \bold{x}_t)$
in (\ref{2.1}) by $(\bQ, \bU\bold{x}_t)$. In the sequel, we always assume
that $\bA$ in (\ref{2.1}) is orthogonal, i.e., $\bA'\bA=\bI_p$, where
$\bI_p$ denotes the $p\times p$ identity matrix. Write \beqn \bA=(\bA_1, \bA_2),\nn\eeqn
where $\bA_1$ and $\bA_2$ are respectively, $p\times (p-r)$ and $p\times
r$ matrices. As now $\bx_{t2} = \bA_2'\by_t$, the linear space spanned by
the columns of $\bA_2$, denoted by
$\mathcal{M}(\bA_2)$, is called the \textsl{cointegration space}. In fact this
cointegration space is uniquely defined by (\ref{2.1}), though $\bA_2$ itself is not.

To highlight the key idea of the new approach, we only consider in this section
and also Section 3 below the cointegration with $\bx_{t2}\sim I(0)$.
The extension of our method to the cases when $\bx_{t2} \sim I(d)$ with $0< d< \min_{1\le j\le p} d_j$ are presented in Section 4 which also allows $d_j$'s and $d$ to be
fractional numbers.

\subsection{Estimation}

The goal is to determine the cointegration rank $r$ in (\ref{2.1}) and to
identify $\bA_2$, or more precisely $\mathcal{M}(\bA_2)$.
Then
$\mathcal{M}(\bA_1)$ is the orthogonal complement of $\mathcal{M}(\bA_2)$,
and $\bold{x}_{it}=\bA'_i\by_t$ for $i=1, 2.$
Our estimation method is motivated by the following observation.
For $j\geq 0$, let
\beqn \hat{\bSigma}_{j}={1\over
n}\sum_{t=1}^{n-j}(\bold{y}_{t+j}-\bar{\bold{y}})(\bold{y}_t-\bar{\bold{y}})',
\quad \quad \bar{\bold{y}}={1\over n}\sum_{t=1}^{n}\bold{y}_t.\nn\eeqn
For any $\ba\in \mathcal{M}(\bA_2), \, \ba'\hat{\bSigma}_{j}\ba$ is the
sample autocovariance function at lag $j$ for the weakly stationary
univariate time series $\ba'\by_t$, and it
converges to a finite constant (i.e. the autocovariance function of $\ba'\by_t$ at lag $j$)
almost surely under some mild conditions. However for any $\ba\notin
\mathcal{M}(\bA_2), \, \ba'\by_t$ is   $I(d)$ for some $d\geq 1$,
and
\beqn\label{2.5}\ba' \hat{\bSigma}_{j} \ba=O_{e}(n^{2d-1}) \quad \,
\hbox{or} \quad \, O_{e}(n^{2d}), \eeqn
depending on whether $E(\ba'\by_t) =0$ or not,
see
Theorems 1 $\&$ 2 of Pe\~{n}a and Poncela (2006).
In the above expression,
   $U=O_{e}(V)$ indicates that $P(0 <  |U/V|<\infty ) \to 1$.
Hence intuitively the
$r$ directions in the cointegration space $\mathcal{M}(\bA_2)$ make
$|\ba'\hat{\bSigma}_{j} \ba|$ as small as possible for all $j\geq 0$.

To combine  information over different lags,  define
\begin{equation} \label{2.6}
\hat{\bW}=\sum_{j=0}^{j_0}\hat{\bSigma}_j\hat{\bSigma}'_j,
\end{equation}
 where $j_0\geq 1$ is  a prespecified and fixed  integer  with respect to $n$ throughout.
We use the product $\hat{\bSigma}_j\hat{\bSigma}'_j$ instead of
$\hat{\bSigma}_j$ to ensure  each term in the sum is non-negative
definite, and that there is no information cancellation over different
lags.
Note that $\ba'\hat{\bSigma}_j\ba=O_e(1)$ if $\ba\in
\mathcal{M}(\bA_2)$,  and is at least
of the  order of  $n^{2d-1}$ if $\ba\in
\mathcal{M}(\bA_1)$,
where $d$ is the minimum integration order of the
components $\bx_{t1}$.
It can be shown that the $(p-r)$ largest eigenvalues of $\wh \bW$ are at
least of the order $n^{2d-1}$, while the other $r$ eigenvalues are
$O_e(1)$ (see (\ref{6.38}), (\ref{6.39}) below).
Hence intuitively $\calM(\bA_2)$ can be estimated by the linear space
spanned by the $r$ eigenvectors of
$\wh \bW$ corresponding to the $r$ smallest eigenvalues, and
$\calM(\bA_1)$  can be estimated by that
spanned by the $(p-r)$ eigenvectors of
$\wh \bW$ corresponding to the $(p-r)$ largest eigenvalues.

Let $(\hat{\bgamma}_1, \cdots, \hat{\bgamma}_p)$
be the orthonormal eigenvectors of $\hat{\bW}$ corresponding to the
eigenvalues arranged in  descending order.
Define
\begin{equation} \label{Ax}
 \hat{\bA}=(\hat{\bA}_1, \hat{\bA}_2), \quad \,
\hat{\bold{x}}_{t1}=\hat{\bA}'_1 \bold{y}_t \quad  \hbox{and}
\quad \wh \bx_{t2}= \hat{\bA}'_2 \bold{y}_t.
\end{equation}
Then $\calM(\wh\bA_1)$ and $\calM(\wh\bA_2)$,  the linear
spaces spanned by the eigenvectors of $\wh \bW$, are consistent estimators for
$\calM(\bA_1)$ and $\calM(\bA_2)$ respectively; see Theorem~\ref{them1} below.

The idea of using an eigenanalysis based on a quadratic form of sample
autocovariance matrices has been used for factor modelling for dimension
reduction (Lam and Yao 2012, and references within), and for segmenting a
high-dimensional time series into several both contemporaneously and
serially uncorrelated subseries (Chang et al. 2017).
One distinctive advantage of using the quadratic form $\wh \bSigma_j \wh \bSigma_j'$ instead
of $\wh \bSigma_j$ in (\ref{2.6})
is that there is no  information cancellation over different lags.
Therefore this approach is insensitive to the choice
of $j_0$ in (\ref{2.6}).
Often small values such as $j_0=5$ are sufficient
to catch the relevant characteristics, as serial dependence is usually
most predominant at small lags.
Using different values of $j_0$ hardly changes the results;
 see Table \ref{tab5} in Section 5 below, and also
Lam and Yao (2012) and Chang et al. (2017).

\subsection{Determining cointegration ranks}

The components of $\wh \bx_t = \wh \bA' \by_t
\equiv (\wh x_t^1, \cdots, \wh x_t^p)'$,
defined in (\ref{Ax}), are arranged according to  descending order of
the eigenvalues of $\wh \bW$.
  Therefore, the order of the components reflects inversely the
closeness to stationarity of the component series, with $\{\wh x_t^p \}$ most
likely being stationary, and $\{\wh x_t^1\}$ most likely
being  $I(d)$ with  largest possible integer $d\ge 1$.
 Let $S_i(m)=\sum_{k=1}^{m}\wh\rho_i(k)$, where
$\wh\rho_i(\cdot)$ is the sample autocorrelation function (ACF) of $\wh x_t^i$ defined as
 $$\wh\rho_i(k)=\Big({1\over {n-k}}\sum_{t=1}^{n-k}(\wh x_{t+k}^i-\overline{\wh x}^i)(\wh x_{t}^i-\overline{\wh x}^i)\Big)
\Big/\Big({1\over n}\sum_{t=1}^{n}(\wh x_{t}^i-\overline{\wh x}^i)^2\Big), \quad
i=1, 2, \cdots, p,$$
where $\overline{\wh x}^i=\sum_{t=1}^{n}\wh x_{t}^i/n.$
When $\wh x_{t}^i$ is stationary and suitable additional conditions hold,  $\lim_{m\rightarrow\infty}S_i(m)<\infty$ in probility, however, when $\wh x_{t}^i$ is non-stationary,
$\wh\rho_i(k)\rightarrow 1$ in probability for any fixed $k$. Hence
$\lim_{m\rightarrow\infty}S_i(m)=\infty.$ Therefore,
 we can estimate the cointegration rank $r$  by
\beqn \label{ACF}\wh r=\sum_{i=1}^{p}I\{S_i(m)/m < c_0\}
\label{new1}
\eeqn
for some constant $0<c_0<1$ and large $m$.  For a
classical
stationary ARMA time series,
the autocorrelation $\rho_i(k)$ decays exponentially, i.e., there
exists a  $\rho\in (0,1)$ such that $\rho_i(k)=O(\rho^k)$.  Hence it is
usually sufficient to use a moderate $m$ in (\ref{new1}).
In our numerical experiments
reported in Section 5, we always set $c_0 =0.3$ and $m=20$, and the
estimator $\wh r$ performs very well and is robust across the different settings.

\begin{remark} For unit-root processes, $\wh r$ defined in (\ref{new1})
typically takes the value 0 with probability  approaching  1.
To appreciate this, let $y_t=y_{t-1}+\varepsilon_t$ be a unit root
process  and $\wh \rho(k)$ be its sample ACF $\wh \rho(k)=\wh \gamma(k)/\wh \gamma(0),$ where
\[\wh \gamma(i)={1\over n}\sum_{t=1}^{n-i}(Y_{t}-\overline{Y})(Y_{t+i}-\overline{Y}), \quad  \overline{Y}=\sum_{i=1}^{n}Y_i/n.\]
Under some regularity conditions on
$\varepsilon_t$, similar to those in Theorem 1 of Bierens (1993), it can be shown that
\beqn \label{remark1}{n\over {m+1}}\Big({{\sum_{{\color{red}k}=1}^m \wh \rho(k)}\over m}-1\Big)\stackrel{d}{\longrightarrow}
-{{(W(1)-\int_{0}^{1}W(t)\, dt)^2+(\int_{0}^{1} W(t)\, dt)^2+d_m}\over {4[\int_{0}^{1}W^2(t)\, dt-(\int_{0}^{1}W(t)\, dt)^2]}},\eeqn
where \[d_m={1 \over \sigma^2}\Big(c(0)+2\sum_{i=1}^{m-1}{{(m-i)(m-i+1)}\over
m(m+1)}c(i)\Big), \quad c(i)=\rm{cov}(\vp_0, \vp_i),   \, \, \,
\, \sigma^2=\lim_{n\rightarrow\infty}{1\over
n}\mathrm{E}\Big(\sum_{s=1}^{n}\varepsilon_s\Big)^2.\] Thus
$\sum_{t=1}^m \wh
\rho(k)/m\stackrel{p}{\longrightarrow} 1$, provided that $n/m$ is large
enough.
\end{remark}

We may also estimate $r$ by  unit-root tests.
For  a given integer $r_0\le 1$,   testing a hypothesis on
 cointegration order
$H_0:  r < r_0$
can be transformed to  testing  a unit-root hypothesis
\begin{align}
 H_0: \wh x_{t}^{p-r_0+1} \sim I(d) \mbox{ for some integer } d \ge 1.
\label{hyper0}
\end{align}
We can apply the test method of Phillips and Ouliaris (1988) to test
 (\ref{hyper0}) as $d$ may be greater than 1. When the null
hypothesis $H_0$ is rejected, we conclude $r$ is at least as large as $r_0$.

\subsection{Estimation for high integration orders}
Let $r_1, \cdots, r_q$ be $q$ positive integers, and $r_1 + \cdots + r_q = p-r$.
Let  $1\le a_1 < \cdots < a_q $ be $q$ integers such that
$\bx_{t1}=(\bx_{t1q}, \cdots, \bx_{t11})=(\bA'_{1q}\by_t, \cdots, \bA'_{11}\by_t)$, where $\bx_{t1j}$ is an $r_j\times 1$ $I(a_j)$ process.
Let
\begin{equation}
\label{A1j}
\wh \bA_1 = ( \wh \bA_{1q}, \cdots, \wh \bA_{11}),
\end{equation}
where $\wh \bA_{1j}$ has $r_{j}$ columns. Then $\wh \bx_{t1j} = \wh
\bA_{1j}' \by_t$ is the estimated component of $\bx_{t1}$  of integration order $a_j$.

Similar to Section 2.3
above, a unit-root test can be adapted to
estimate the sizes $r_1, \cdots, r_q$ and the integration
orders $a_1, \cdots, a_q$.
  We illustrate the
idea below by outlining the steps in estimating $(a_1, r_1)$, they
can be repeated in order to estimate $(a_2, r_2), (a_3, r_3), \cdots$.

For $\wh r$ defined in (\ref{ACF}), let $\wh a_1$ be the minimum integer $d\ge 1$ such
that a unit-root test rejects
$
H_0: \nabla^d \wh x_{t}^{p-\wh r}\, \sim I(1)$ against  $H_1: \nabla^d
\wh x_{t}^{p-\wh r}\, \sim I(0)$.
Then the size $r_1$ can be estimated by applying estimator (\ref{ACF}) to
the $(p-\wh r)\times 1$ series
$\{ \nabla^{\wh a_1} \wh x_t^{j}, \, j=1, \cdots, p- \wh r\}$.

\section{Asymptotic Properties}
\setcounter{equation}{0}

In this section, we investigate the asymptotic properties of  the
proposed statistics.
First, we show that with  $r$ given,
the linear space ${\mathcal{M}}(\wh\bA_2)$
consistently estimate   the cointegration space $\mathcal{M}(\bA_2)$.
We measure the distance between the two spaces by
\begin{equation} \label{3.1}
D(\mathcal{M}(\wh\bA_2), \mathcal{M}(\bA_2))=\sqrt{1-{1\over
r}\mathrm{tr}(\hat{\bA}_2\hat{\bA}'_2 \bA_2 \bA'_2)}.
\end{equation}
Then $D(\mathcal{M}(\wh\bA_2), \mathcal{M}(\bA_2)) \in [0, 1]$, being 0
if and only if ${\mathcal{M}}(\wh\bA_2) ={\mathcal{M}}(\bA_2)$,
and 1 if and only if ${\mathcal{M}}(\wh\bA_2)$ and ${\mathcal{M}}(\bA_2)$ are orthogonal.
Furthermore, we show that  the estimator
$\wh r$, defined in (\ref{ACF}), is
consistent.
We consider two asymptotic regimes: (i) $p$ is fixed while $n\to \infty$,
and (ii) $p \to \infty$ more slowly than  $n$.

Put $\bx_{t1} = (x_{t}^1, \cdots, x_{t}^{p-r})'$. Under (\ref{2.1}),
 $x_{t}^j$ is  $I(d_j)$ for $1\le j \le p-r$ and  $z_{t}^j \equiv\nabla^{d_j} x_{t}^j$ is
I(0), where $d_j \ge 1$ is an integer.
Write $\bz_t = (z_{t}^1, \cdots, z^{p-r}_t)'$ and
$\bve_t = (\bz_t', \bx_{t2}')'$. Denote  the vector of partial sums of components  of $\bve_t$  by
\beqn \bS_n(\bold{t})\equiv (S_n^1(t_1), \cdots, S_n^p(t_p))'
=\Big({1\over
\sqrt{n}}\sum_{l=1}^{[nt_1]}(\varepsilon_l^1-\mathrm{E}\varepsilon_1^1),
\cdots, {1\over
\sqrt{n}}\sum_{l=1}^{[nt_p]}(\varepsilon_l^p-\mathrm{E}\varepsilon_1^p)\Big)',\nn
\eeqn
where $0< t_1 < \cdots < t_p \le 1$ are constants and $\bt =(t_1, \cdots, t_p)'$.

\subsection{When $n \to \infty$ and $p$ is fixed}
\label{sec31}

We introduce a regularity condition first.

\vskip3mm

 {\bf Condition 1}.
 \begin{itemize}
\item[(i)] There exists a Gaussian process $\bold{W}(\bt)=(W^1(t_1),
\cdots, W^p(t_p))'$ such that as $n\rightarrow\infty$,
\beqn \bS_n(\bold{t})\stackrel{J_1}{\Longrightarrow} \bold{W}(\bt), \quad \hbox{on} \, \, \,  D^p(0, 1),\nn\eeqn
 where $\stackrel{J_1}{\Longrightarrow}$ denotes  weak convergence
under Skorohod $J_1$ topology (Chapter 3 in Billingsley 1999),
and  $\bold{W}(\bold{1})$ has a positive definite
covariance matrix $\bOmega=(\sigma_{ij}).$

\item[(ii)] The sample autocovariance matrix of $\bx_{t2}$ satisfies
\beqn \max_{0\leq j\leq j_0}\big\|{1\over n}\sum_{t=1}^{n-j}(\bx_{t+j, 2}
-\bar{\bx}_2)(\bx_{t2}-\bar{\bx}_2)'- \hbox{Cov}(\bx_{1+j,2}, \bx_{1,2})\big\|_2
\stackrel{p}{\longrightarrow} 0,\nn\eeqn
where $\|\bH\|_2=\max_{\|\ba\|=1}\|\bH \ba\|$ is the $L_2$-norm of matrix $\bH$, $\bar{\bx}_2$ is the sample mean of $\bx_{t2}$,
and $\stackrel{p}{\longrightarrow}$ denotes  convergence in probability.

\end{itemize}

Note that our definition of cointegration is formally different from that of
Johansen (1995) which is based on ARIMA framework.
There are some subtle technical differences between the respective
conditions. For example, Condition 1(i) above implies det($\var(\bve_t))\ne 0$
while Johansen's setting allows the ARIMA process driven by a degenerate
innovation process.

In fact, Condition 1 is mild.
It is fulfilled when $\{\bve_t\}$ is weakly stationary with det$(\var(\bve_t))\ne 0$,
$\mathrm{E}\|\bolde_t\|^{2\gamma}< C $ for some constants $\gamma>1$ and
$C< \infty$, and $\{\bve_t\}$ is
also $\alpha$-mixing with
 mixing coefficients $\alpha_m$ satisfying the condition
$\sum_{m=1}^{\infty}\alpha_{m}^{1-1/\gamma}<\infty$; see
Theorem 3.2.3 of Lin and Lu (1997). It is also fulfilled when
$\bolde_t=\sum_{j=0}^{\infty}\bC_j\boldeta_{t-j}$,
where $\bfeta_t$ are i.i.d. with  non-singular covariance matrix and
$E\|\bfeta_t\|^{4\ga} < \infty$ for some constant $\ga>1$,
and  det$( \sum_{j=0}^\infty \bC_j) \ne 0$,
$\sum_{j=1}^\infty  ||\bC_{j}|| < \infty$. See Fakhre-Zakeria and  Lee (2000).
Note that our setting accommodates the cases when $\by_t$ contains
linear deterministic components,
as we allow  $\mathrm{E}(\bve_{t}) \ne 0$.

 \begin{theorem} \label{them1}
Let $r$ be known.
 Under Condition 1, $D(\mathcal{M}(\wh\bA_2), \mathcal{M}(\bA_2)) = o_p(1)$.
Furthermore,
\begin{itemize}
\item [(i)] $D(\mathcal{M}(\wh\bA_2), \mathcal{M}(\bA_2))=O_e(n^{-2a_1+1})$ provided either (a)
$|I_0|\geq 2$ or (b) $ |I_0|=1$ and
$\mathrm{E}z_t^{I_0}= 0$,
\item [(ii)] $D(\mathcal{M}(\wh\bA_2), \mathcal{M}(\bA_2))=O_e(n^{-2a_1})$
provided $|I_0|=1$ and
$\mathrm{E}z_t^{I_0}\neq 0$, and
\item [(iii)]
$D(\mathcal{M}(\wh\bA_{1j}),
\mathcal{M}(\bA_{1j}))=O_e(n^{-2\alpha_j})$ for $j=1, \cdots, q$
provided $\mathrm{E}\bz_t=0,$
 \end{itemize}
where  $ I_0=\{i: x_t^i\sim I(a_1), \, 1\le i \le p-r\}, \, |I_0|$ denotes the number of  elements  in $I_0$,
$\alpha_j=\min\{a_j-a_{j-1}, a_{j+1}-a_{j}\}, \, a_0=1/2$ and
$ a_j, \, j= 1, \cdots, q$ are defined  in Section 2.4.
 \end{theorem}

\begin{remark} When $\mathrm{E}\bz_t\neq 0$, we can express the
components $x_t^i$ of $\bx_{t1}$ as
 \[ (1-B)^{d_i}x_t^i= (z_t^i-\mathrm{E}z_t^i)+\mathrm{E}z_t^i=:\varepsilon_t^i+\mu_i.\]
Hence
 \[x_t^i=(1-B)^{-d_i}\varepsilon_t^i+\mu_i\prod_{l=0}^{d_i-1}(t+l)/(d_i!)
 =:\xi_t^i+\mu_i\prod_{l=0}^{d_i-1}(t+l)/(d_i!).\]
This entails
 $\by_t=\bA \bx_t=\bA (\boldxi'_t, \bx'_{t2})'+ \bB (1, t, t^2, \cdots,
t^{a_q})'$, where $\bB$ is a  $p\times a_q$ matrix. We can estimate
$\bB$ by the least squares method based on $\{\by_t\}$,
and identify the cointegration subspaces spanned by
$\bA_{1j}$ using the
detrending series $\widetilde{\by}_t=\by_t-\widehat{\bB} (1, t, t^2,
\cdots, t^{a_q})'$.
It can then be shown that Theorem 1 (iii) still holds.
 \end{remark}

\begin{theorem}
Under Condition 1,
$\lim_{m\rightarrow\infty}P(\, \wh{r}=r\, )= 1$.
 \end{theorem}

\subsection{When $n\to \infty$, $p\to \infty$ and $ \ p=O(n^c)$}
\label{sec32}

We extend the asymptotic results in the previous section to the cases when
$p\to \infty$ and $p=O(n^c)$ for some $c \in (0, 1/2)$. Technically we
employ a normal approximation
method to establish the results. See Condition 2(i) below.

\vskip3mm

{\bf Condition 2.}
 \begin{itemize}
\item[(i) ] 
Suppose that  there exists an $m$-dimensional vector  $\bfe_t$
with mean zero and independent components such that $\bz_t=\bB\bfe_t$,
where $\bB$ is a $(p-r)\times m$ matrix, $m\geq p-r$ and $\|\bB\|_2<\infty$.
For each component $e_t^i$ of $\bfe_t$, there exists an independent and
standard normal sequence $\{\nu_{t}^i\}$
for which as $n\rightarrow\infty$,
    \beqn \label{3.9}\max_{1\leq i\leq m}\max_{0\leq t\leq 1}\mathrm{E}\Big[\sum_{s=1}^{[nt]}(e_{s}^{i}
    -\sigma_{ii}\nu_{s}^i)\Big]^2=O(n^{2\tau}),  \eeqn
    where $0<\tau<1/2$ is a constant,
 $b_1\leq \sigma_{ii}^2 \equiv
\lim_{n\rightarrow\infty}\mathrm{Var}\left(\sum_{s=1}^{n}e_{s}^{i}\right)/n\leq b_2$ for any $i$, and $b_1, \, b_2$ are two positive constants.

\item[(ii) ] 
The sample autocovariance matrix of $\bx_{t2}$ satisfies
\beqn \max_{0\leq j\leq j_0}\Big\|{1\over n}\sum_{t=1}^{n-j}(\bx_{t+j,
2}-\bar{\bx}_2)(\bx_{t2}-\bar{\bx}_2)'- \hbox{Cov}(\bx_{1+j,2},
\bx_{1,2})\Big\|_2\stackrel{p}{\longrightarrow} 0.\nn\eeqn

\item[(iii) ] 
Suppose $\{\bz_t\}$ and $\{\bx_{t2}\}$ are independent and for $\tau$ given above
    \beqn \max_{p-r< j\leq p}\sum_{s, t=1}^{n}|\mathrm{E}(\varepsilon_s^j\varepsilon_{t}^j)|
    =O(n^{1+2\tau}).\nn\eeqn

\end{itemize}

\begin{remark} The inequalities immediately below (3.2)
holds when all components series of $\bz_t$ are I(0) with spectral
density continuous at zero
frequency. This is guaranteed by the fact that
 their variance  is proportional to the Cesaro sum of the
Fourier series of the spectral density at zero frequency, and thus
converges to the latter (which is positive and finite under I(0)) after
normalization.
\end{remark}

\begin{remark}  The form $\bz_t=\bB\bfe_t$ in Condition 2 (i), has  been
used by Bai and Saranadasa (1996) and Chen and Qin (2010).
Many classic vector time series  including stationary VAR, VARMA  and more generally the linear process
\beqn \bz_t=\sum_{j=0}^{\infty}\bB_{j}\bfe_{t-j}\nn\eeqn
with $\sum_{j=0}^{\infty}\|\bB_j\|_2<\infty$
follow this from.
We require  $m \geq p-r$, which ensures that no linear combination of
$\bz_t$ is $I(0)$.
The assumption on the independence between  $\{\bz_t\}$ and
$\{\bx_{t2}\}$ in Condition 2(iii) ensures that  cross
correlation of $\{\bz_t\}$ and $\{\bx_{t2}\}$ is negligible in deriving
 the properties of the  eigenvalues of $\wh\bW$, which can be replaced by the
condition
that $\mathrm{E}(n^{-(d_i+1/2)}\sum_{t=1}^{n}x_{t}^i x_{t}^h)^2=o(1/(pr))$.
\end{remark}

\begin{remark}
Let $p=o(n^{1/2}).$ Condition 2 is implied by any of the three assertions below.

\begin{enumerate}
\item[(i)]
The components of $\bolde_t$ are independent of each other, and each component
series  $\{\varepsilon_{t}^{i}\}$ is a martingale difference sequence
with $\max_{1\leq i\leq p}\mathrm{E}|\varepsilon_{t}^{i}|^{q}<\infty $ for some $ q>2$. Furthermore, for some $2<q^{*}\leq \min\{4, q\},$
\beqn\max_{1\leq i\leq p}\mathrm{E}\left|\sum_{t=1}^{n}[(\varepsilon_{t}^{i})^2-\sigma_{ii}^2]\right|
=O(n^{2/q^{*}}).\nn\eeqn

\item[(ii)]
The components of $\bolde_t$ are independent,  $\mathrm{E}\bolde_t=0$, and
$\displaystyle\max_{1\leq i\leq p}\mathrm{E}|\varepsilon_t^i|^{\kappa}<\infty$ for some $\kappa>q \in (2, 4]$.
The process $\{\bolde_t\}$ is  $\alpha$-mixing with  mixing coefficients $\alpha_m$
satisfying
\beqn \label{18}\sum_{m=1}^{\infty}\alpha_m^{(\kappa-q)/(\kappa q)}<\infty.\eeqn

\item[(iii)]
The components of $\bolde_t$ are independent.  Each component
$\varepsilon_t^i$ satisfies the following conditions.
\begin{itemize}
\item[(a)] There exists an i.i.d random sequence $\{\eta_t^i\}$ such that
\beqn \label{16}\varepsilon_t^i=\sum_{j=0}^{\infty}c_{ij}\eta_{t-j}^i.\nn\eeqn

\item[(b)] $\mathrm{E}\varepsilon_t^i=0, \, \mathrm{E}|\varepsilon_t^i|^{q}<\infty$ for some $q>2$ and $\sum_{j=0}^{\infty}j|c_{ij}|<\infty.$
\end{itemize}
\end{enumerate}
 \end{remark}

\begin{theorem}\label{them31} Let $r$ be known and Condition 2 hold.
If $p=o(n^{1/2-\tau})$ and $\tau$  given in Condition 2,
 $$D(\mathcal{M}(\wh\bA_2), \mathcal{M}(\bA_2))=
O_p(p^{1/2}n^{-2a_1+1}(\lambda^*)^{-1}) 
,$$
 where $\lambda^*$ is the smallest eigenvalue of
$\int_{0}^{1}\bold{F}(t)\bold{F}'(t)\,dt$ defined in Lemma 9 in Section \ref{sec7} below.
\end{theorem}

\begin{remark} Theorem \ref{them31} is derived under the condition
$p=o(n^{1/2-\tau})$, while there are no direct constraints on either $r$ or $p-r$.
However when $p-r$ is fixed, $\int_{0}^{1}\bold{F}(t)\bold{F}'(t)\,dt$ is a
$(p-r)\times (p-r)$ positive definite matrix, and, hence,
$\lambda^*$ is  positive and  $O_e(1)$.  When  the  integration orders of
all the nonstationary components are the same and equal to $d_{min}$, then
$(\lambda^*)^{-1}=O_{p}((p-r)^{2d_{min}-1}).$

\end{remark}

\begin{theorem} Let Condition 2 hold  and $p=o(n^{1/2-\tau})$.
Then
 $$\lim_{n\rightarrow\infty}P(\, \wh{r}=r\, )= 1,$$ provided $(\lambda^*)^{-1}p^{1/2}n^{-a_1+1/2}=o(1)$.
\end{theorem}

\section{Fractional cointegration}
\label{sec5}

Fractional cointegration has attracted increasing attention  in recent years, see, e.g.,
 Robinson and Hualde (2003), Chen and Hurvich (2006) and
Robinson (2008). In this section, we generalize the method
presented in Section 2 to  cases when the components
of $\by_t$ may be fractionally integrated.
 For simplicity, we now assume $p$ is fixed. 

Let $v_{t}^{+}=v_t \idv(t>0)$ and for any $\alpha \in \mathbb{R}$,
\beqn \Delta^{-\alpha}=\sum_{j=0}^{\infty}a_j(\alpha)B^j, \quad a_j(\alpha)={\Gamma(j+\alpha)\over {\Gamma(\alpha)\Gamma(j+1)}}\nn\eeqn
be formally defined as in Hualde and Robinson (2010), where $B$ is the backshift operator.
With these definitions we can extend the definition of the $I\left(
d_{1},...,d_{m}\right) $ process $\mathbf{v}_{t}$ in Section 2 to
non-negative real-valued $d_{i},$ such that $d_{i}\neq k-1/2$ for any
integer $k.$  Note that for $d_{i}<1$/$2$ the $i$th element of $\mathbf{v}%
_{t}$ is `asymptotically stationary' (due again to the truncation in the
definition of $\mathbf{v}_{t}$), while $d_{i}>1$/$2$ represents the
`nonstationary' region.

With this extended definition to cover fractional time series we again
consider a $p\times 1$ \ observable $I\left( d_{1}, \cdots,d_{p}\right) $ time
series  $\mathbf{y}_{t}$ satisfying (2.1), partitioning $\mathbf{x}_{t}$ as
before.  However we also extend the definition of cointegration, saying
that  $\mathbf{y}_{t}$ is cointegrated if at least two $d_{i}$ are equal
and exceed $1$/$2$ and there exists a linear combination giving nonzero
weight to two or more of these that is $I\left( c\right) $ $\ $for $0\leq
c< d_i.$  Thus, let  $a_1>1/2$ be the smallest integration order of
elements of $\mathbf{x}_{t1}$ and let $ \delta \in \left[ 0, a_1\right) $ be
the  integration order of elements of $\mathbf{x}_{t2}.$
 Thus, each component of  $\bx_{t2}$ is a cointegrating error
of $\by_t$. Let $\bA=(\bA_1, \bA_2)$ and $\mathcal{M}(\bA_2)$ be defined as in
Section 2. Then  $\mathcal{M}(\bA_2)$ is called the fractional
cointegration space  and  $r$ is called the fractional cointegration rank.
We estimate $\mathcal{M}(\bA_2)$ and $r$ in the same manner as in Section 2.

Furthermore, let  $r_1, \cdots, r_q$ be $q$ positive integers with $r_1 +
\cdots + r_q = p-r$, and
  $1/2 < a_1 < \cdots < a_q $.  Suppose that $\bx_{t1}$ consists  of
 $r_j$ $I(a_j)$ components.
Let
\begin{equation}
\label{FA1j}
\wh \bA_1 = ( \wh \bA_{1q}, \cdots, \wh \bA_{11}),
\end{equation}
where $\wh \bA_{1j}$ has $r_{j}$ columns. Then $\wh \bx_{t1j} = \wh
\bA_{1j}' \by_t$ is the estimated components  of $\bx_{t1}$ (i.e., $\bx_{t1j}=\bA'_{1j}\by_t$)  of integration order $a_j$.

Let $\bve_i=(\varepsilon_{i}^1, \cdots, \varepsilon_i^p)'$ be the $p$-dimensional $I(0)$ with mean zero such that $\nabla^{d_j}x_i^j=\varepsilon_i^{j}+\mu_j.$ Let $
 \bold{S}_{n}(t)=\sum_{i=1}^{[nt]}\bve_i$ and $
 I_1=\{i: d_i<1/2, \, 1\leq i\leq p\}$.

\vskip3mm
 {\bf Condition 3.}

\begin{itemize}

\item[](i) $\mathrm{E}||\bve_{t}||_2^q<\infty$ for some $q>\max(4, 2/(2a_1-1))$ and for any $i, j\in I_1$, as $n\rightarrow\infty,$
    \beqn {1\over n}\sum_{t=1}^{n}x_t^i x_t^j \stackrel{p}{\longrightarrow} \mathrm{E}[x_1^i x_1^j].\nn
    \eeqn

\item[](ii) There exists an i.i.d mean zero $p\times 1$ normal vector $\{\bold{w}_i\}$ such that as $n\rightarrow\infty,$
\beqn \max_{0\leq t\leq 1} ||S_{n}(t)-\sum_{i=1}^{[nt]}\bold{w}_i||_2=o_p(n^{1/s}), \, \, \hbox{for some} \ \ s>2.\nn\eeqn

\end{itemize}

 \begin{remark}  Condition 3 is mild and satisfied by either of the following processes.
\begin{enumerate}
\item Suppose $\bve_t$ follows a linear process:
\beqn \bve_t=\sum_{k=0}^{\infty}\bold{C}_{k}\bold{e}_{t-k}, \, t=1, 2, \cdots\nn\eeqn
and $\{\bold{e}_{t}\}$ are i.i.d  vectors with mean zero, $\mathrm{E}\bold{e}_{t}\bold{e}'_{t}=\Sigma_e>0, \, \mathrm{E}||\bold{e}_{t}||_2^q<\infty$ for some $q>4,$
the $p\times p$ coefficient matrices $\bold{C}_{k}$ satisfy $\displaystyle\sum_{k=0}^{\infty}k||\bold{C}_{k}||^2<\infty.$  Then, by Lemma 2 of Marinucci and Robinson (2000), we have (ii) of Condition 3 holds. (i) follows by  ergodicity.

\ignore{\item Suppose $\bve_t$ is weakly stationary with det$(\var(\bve_t))\ne 0$,
$\mathrm{E}\|\bolde_t\|^{\kappa}< C $ for some constants $\gamma>1$ and
$C< \infty$, and is
also $\alpha$-mixing  with
 mixing coefficients satisfying the condition (\ref{18}), then (ii) of Condition 3 follows from (\ref{19}). Using (\ref{18}) and the element inequality of mixing process (see Lemma 1.2.4 of Lin and Lu (996)), it can be shown that
\beqn \mathrm{E}\left|{1\over n}\sum_{t=1}^{n}[x_t^{i}x_t^j-\mathrm{E}(x_t^{i}x_t^j)]\right|&=&
{1\over n}\mathrm{E}\left|\sum_{t=1}^{n}\sum_{h_1=0}^{t}\sum_{h_2=0}^{t}
\alpha_{t-h_1}(d_i)\alpha_{t-h_2}(d_j)[\varepsilon_{h_1}^i\varepsilon_{h_2}^j
-\mathrm{E}(\varepsilon_{h_1}^i\varepsilon_{h_2}^j)]\right|\nn\\
&=&o(1),\nn\eeqn
which implies (i) of Condition 3.}

\item Suppose $\bve_t$ follows a generalized random coefficient autoregressive model:
\beqn \bve_t=\bC_t \bve_{t-1}+\bold{e}_t\eeqn
and $\{(\bC_t, \bold{e}_t)\}$ are i.i.d random variables with $\mathrm{E}||\bC_1||_2^q<1$ and $\mathrm{E}||\bold{e}||^q<\infty$ for some $q>2$, then (ii) of Condition 3 holds with
$s<\min\{q, 4\}$, see Corollary 3.4 of Liu and Lin (2009). Similarly, (i) follows by  ergodicity.

\end{enumerate}

 \end{remark}
\begin{theorem} \label{them5}
Let $r$ be known.
 Under Condition 3, $D(\mathcal{M}(\wh\bA_2), \mathcal{M}(\bA_2)) = o_p(1)$.
 Furthermore,
\begin{itemize}
\item [(i)] when $\delta<1/2$,
\begin{itemize}
\item[(a)]$D(\mathcal{M}(\wh\bA_2), \mathcal{M}(\bA_2))=O_e(n^{-2a_1+1})$ provided  either
$|I_0|\geq 2$ or $ |I_0|=1$ and $\mu_{I_0}= 0;$
\item [(b)] $D(\mathcal{M}(\wh\bA_2), \mathcal{M}(\bA_2))=O_e(n^{-2a_1})$
provided $|I_0|=1, \, \mu_{I_0}\neq 0$;
\end{itemize}
\item [(ii)] when $\delta>1/2$ and $\mu_j=0$ for $j\geq p-r$,
$D(\mathcal{M}(\wh\bA_2), \mathcal{M}(\bA_2))=O_e(n^{-2(a_1-\delta)})$;
\item [(iii)] when $\mu_j=0$ for  $j=1, \cdots, p-r$,
 \[D(\mathcal{M}(\wh\bA_{1j}),
\mathcal{M}(\bA_{1j}))=O_e(n^{-2\alpha_j})\quad \hbox{for} \quad j=1, \cdots, q,\]
 \end{itemize}
 where $ I_0$ and $\alpha_j$  are defined as in Theorem 1.
\end{theorem}

\begin{theorem} \label{them6}
Let  Condition 3 hold.
Then
 $\lim_{n\rightarrow\infty}P(\, \hat{r}=r\, )= 1$,
provided $1\le r <p$.
 \end{theorem}

 \ignore{\begin{remark} The above theorems include the most popular
case $\delta<1/2$. Under this special case, $\bx_{t2}$ is the
stationary  part of $\bx_t$ and $\mathcal{M}(\bA_2)$ is  the linear space
spanned by the corresponding columns of $A$. So, the above theorem can
be seen as a generalization of the results in Section 3.1.
 \end{remark}}

\section{Numerical properties}
\setcounter{equation}{0}

We illustrate the proposed method with 4 simulated examples and one real data set.
Note that the comparison with Johansen's (1991) likelihood method is carried
out for Example 1 and the real data example only, as Examples 2 concerns different
integration orders for different components, Example 3 illustrate the method
in the presence of an additional
deterministic linear trend, and Example 4 is a model
of fractional cointegration. Johansen's method is not applicable to \linebreak Examples~2-4.

\noindent
{\bf Example 1.} Let the first three components of $\by_t$ be the same as
Exercise 3.1 in Johansen (1995), i.e.
\beqn \label{5.3} \left(\begin{array}{c} y_{t1}\\
y_{t2}\\
y_{t3}
\end{array}\right)=
\left(\begin{array}{ccc} 1 & 1&0\\
1/2 &0 &1\\
0& 1 &0
\end{array}\right)\left(\begin{array}{c} x_{t1}\\
x_{t2}\\
x_{t3}
\end{array}\right)=:\bA_{11}\left(\begin{array}{c} x_{t1}\\
x_{t2}\\
x_{t3}
\end{array}\right)
,\nn\eeqn
where  $x_{t1}$ is an $I(1)$ process, $x_{t2}, x_{t3}$ and the
innovations in $x_{t1}$ are independent $N(0,1)$. For  $p>3$,
we add to $y_{t1}, y_{t2}, y_{t3}$ above $r-2$ extra   stationary
AR(1) components and $p-r-1$ ARIMA(1,1,1)
components. All the coefficients in AR(1) are 0.5,
the coefficients in ARIMA(1,1,1) are $(0.6, 0.8)$, and all the innovations are
independent $N(0,1)$.
Except for the elements in $\bA_{11}$ specified above, all
the other elements of $\bA$ are generated independently from $U(-3, 3)$.
 For each setting with different combinations of $p, r$ and $n$ (see Table  \ref{tab1}),
we draw 500 samples.
We set $j_0=5$ in (\ref{2.6}), and
estimate the cointegration rank $r$ by (\ref{ACF}) with $c_0=0.3$ for each of the 500 samples.
Then with $r=\wh r$, we estimate $\wh \bA$ by (\ref{Ax}).
Since $\wh r$ is not
necessarily equal to $r$, and $\bA $ is not a half orthogonal matrix (as specified above), we
extend the definition of discrepancy measure (\ref{3.1}) as follows:
\begin{equation} \label{4.1}
D_1(\calM(\wh \bA_2), \calM(\bB_2) ) = \Big\{ 1 -
 {\tr\big(\wh\bA_2 \wh\bA_2'
 \bB_2( \bB_2' \bB_2)^{-1}  \bB_2'\big) \over
\max( r, \wh r)} \Big\}^{1/2},
\end{equation}
where 
$\bB_2$ is the $p \times r$ matrix
consisting of the last $r$ columns of $(\bA^{-1})'$, as now $\bx_{t2}=\bB_2\by_t.$
Then $D_1(\calM(\wh \bA_2), \calM(\bB_2) ) \in [0, 1]$,
being 1 if and only if $\calM(\wh \bA_2)$ and $\calM(\bB_2)$ are mutually orthogonal, and 0 if and only if the two subspaces are the same. When
$\wh r=r$ and $ \bA' \bA=\bI_p$, $\bB_2 = \bA_2$ and
 $D_1(\calM(\wh \bA_2), \calM(\bB_2) ) =
D(\calM(\wh \bA_2), \calM(\bA_2) )$ defined in (\ref{3.1}).
The relative frequencies (RF) for the occurrence of the event $\{ \wh r = r\}$ and
the average value of $D1=D_1(\calM(\wh \bA_2), \calM(\bB_2) )$ over 500 replications
 are listed in Table  \ref{tab1} under the name new method (New).

Also included in Table  \ref{tab1} are the results of Johansen's
likelihood estimation with  cointegration rank $r$ estimated by the trace test; see
Johansen (1991). We apply the method twice with  testing level set
at 0.05 and 0.01, respectively, marked as Jo(0.05) and Jo(0.01) in Table  \ref{tab1}.
The null-distribution of the trace test statistic is approximated by
that of
\[
\Big[\sum_{t=1}^{T}\bolde_t(\boldX_{t-1}-\bar{X})'\Big] \Big[\sum_{t=1}^{T}(\boldX_{t-1}-\bar{\boldX})(\boldX_{t-1}-\bar{\boldX})'\Big]^{-1}
\Big[\sum_{t=1}^{T}(\boldX_{t-1}-\bar{\boldX})\bolde'_t\Big], 
\]
where $\bolde_t=(\varepsilon_{t,1}, \cdots, \varepsilon_{t,p-r})',
\boldX_{0}=0$ and $\boldX_{t}=\sum_{j=1}^{t}\bolde_t$,  and
$\{\varepsilon_{t, i}\}$ are independent $N(0,1)$. See Johansen and
Juselius (1990). This approximate distribution is calculated by simulation
with $T=1000$ and 6000 replications.

Table \ref{tab1} indicates clearly that the newly proposed method always outperforms
Johansen's method. More precisely the estimator $\wh r$ defined in (\ref{ACF})
achieves higher relatively frequencies for hitting the true value $r$ than those
achieved by the trace test with significance level at either 0.05 or 0.01.
Note that the first part of Table \ref{tab1} with $p=3$ and
$r=2$ corresponds to the same setting of Example 3 of Johansen (1995).
The inference is more challenging when $p$ and $r$ increase.
When $p=30, r=10$, our new method works reasonably well when the sample size $n=1000$
and it works almost perfectly when $n \ge 1500$. On the other hand, Johansen's
method, which is not designed for large $p$, fails to perform even when
 $n=2000$ or $2500$.

\renewcommand{\baselinestretch}{1}
\begin{table}[hptb]
\centering
\caption{\small  Relative  frequencies (RF) of  correct estimation of  $r$ and
 average distance $D_1$ defined in (\ref{4.1}) in simulation
with 500 replications for Example 1.} \label{tab1}
\vskip3mm
{\small
\begin{tabular}{|l|l|c|c|c|c|c|c|c|c|c|c|c|c|}
\hline
 &  &\multicolumn{2}{c|}{$n$=200}&\multicolumn{2}{c|}{$n$=300}
 &\multicolumn{2}{c|}{$n$=500}& \multicolumn{2}{c|}{$n$=1000} &
\multicolumn{2}{c|}{$n$=1500}& \multicolumn{2}{c|}{$n$=2000} \\
 \cline{3-14}
 & Method &RF& $D_1$ &RF& $D_1$ &RF& $D_1$ &RF& $D_1$ &RF& $D_1$ &RF& $D_1$  \\
 \hline
 \hline
{$p$=3} &Jo(0.05) &.930 &.051  &.964 &.028 
 &.944 &.036 &.954 &.028 &.942 &.034 &.966 &.020\\
 \cline{2-14}
{$r$=2} &Jo(0.01) &.980 &.026  &.996 &.011 
& .990 &.011 &.992 &.007 &.994 &.005 &.984 &.010\\
\cline{2-14}
 & New &.968 &.032 &.998 &.009 
 &1.00 &.005 &1.00 &.002 &1.00 &.002 &1.00 &.001  \\
 \hline
 \hline
{$p$=6} &Jo(0.05) &.558 &.276  &.636 &.226  
&.644 &.217 &.640 &.215 &.702 &.177 &.640 &.214\\
 \cline{2-14}
{$r$=2} &Jo(0.01)  &.760 &.184  &.856 &.117  
&.802 &.123 &.862 &.083 &.852 &.088 &.866 &.079 \\
\cline{2-14}
 & New   &.388 &.375  &.838 &.117  
 &.982 &.027 &.994 &.013 &1.00 &.004 &1.00 &.006\\
 \hline
\hline
{$p$=9} &Jo(0.05) &.200 &.445   &.216 &.422   
&.312 &.367 &.344 &.345 &.352 &.337 &.380 &.323\\
 \cline{2-14}
{$r$=3} &Jo(0.01)  &.558 &.290   &.598 &.251   
 &.666 &.185 &.708 &.154 &.742 &.135 &.752 &.129 \\
\cline{2-14}
 & New  &.016 &.605   &.384 &.341   
 &.922 &.066 &.998 &.018 &1.00 &.010 &1.00 &.006 \\ 
 \hline
 \hline
{$p$=12}&Jo(0.05) &.064 &.539  &.144 &.466   
&  .198 &.416 &.254 &.375 &.270 &.362 &.288 &.352 \\ 
 \cline{2-14}
{$r$=4} &Jo(0.01)&.226 &.426  &.318 &.354   
& .432 &.282 &.490 &.243 &.520 &.225 &.544 &.212\\ 
\cline{2-14}
 & New  &0    &.681  &.054 &.534   
 &.794 &.120 &.996 &.021 &1.00 &.011 &.998 &.009\\ 
 \hline
 \hline
{$p$=18} &Jo(0.05) &0    &.653   &.006 &.586     
& .016 &.535 &.056 &.478 &.090 &.448 &.092 &.443\\ 
 \cline{2-14}
{$r$=6} &Jo(0.01) &.006 &.595   &.020 &.522     
&.046 &.468 &.158 &.379 &.226 &.349 &.236 &.333\\ 
\cline{2-14}
 & New  &0    &.737   &0    &.675     
 &.092 &.429 &.986 &.032 &1.00 &.016 &1.00 &.011\\ 
 \hline
 \hline
{$p$=24} &Jo(0.05) &0 &.742   &0 &.664  
& 0 &.580 &.008 &.507 &.002 &.488 &.006 &.480\\ 
 \cline{2-14}
{$r$=8} &Jo(0.01) &0 &.703   &0 &.613  
& 0 &.532 &.006 &.468 &.026 &.438 &.020 &.435\\ 
\cline{2-14}
 & New  &0 &.759   &0 &.719  
 & 0 &.593 &.898 &.064 &1.00 &.022 &1.00 &.014 \\ 
\hline
\hline
{$p$=30}&Jo(0.05) &0 &.790   &0 &.732   
&0 &.628 &0 &.556 &0 &.527 &.002 &.512 \\ 
 \cline{2-14}
{$r$=10} &Jo(0.01) &0 &.772   &0 &.691   
& 0 &.591 &.004 &.514 &.004 &.480 &.004 &.466 \\ 
\cline{2-14}
 & New  &0 &.771   &0 &.742   
 &0 &.662 &.482 &.186 &.984 &.030 &1.00 &.018 
\\
 \hline
\end{tabular}
}
\end{table}

\vskip3mm

\noindent
{\bf Example 2.} Now
in model (\ref{2.1}) let $\bx_{t2}$ consist of   $r$
stationary AR(1) processes with coefficients $-0.4+i/r \, (i=1,
\cdots, r)$, and let $s$ components of $\bx_{t1}$ be ARIMA(1,1,1) with  coefficients
$0.3+0.5i/s$ and $0.2+0.6i/s \, (i=1,  \cdots, s)$,
and the other $p-r-s$ components be ARIMA(0,2,1) with
 coefficients generated
independently from $U(-0.95, 0.95)$. Hence $\bx_{t1}$ consists of
a mixture of $I(1)$ and $I(2)$ processes.
All  innovations involved are independent $N(0,1)$.
Let the elements of $\bA$ be generated independently from $U(-3, 3)$.
We estimate the cointegration rank $r$ by (\ref{ACF}), and apply the same method
to the differenced $\wh \bx_{t1}$ to estimate $s$; see Section 2.4 above.
For each setting, we replicate the exercise 500 times. The relative frequencies for
the occurrence of events $\{ \wh r = r\}$ and $\{ \wh s = s\}$ are
listed in
Table \ref{tab2}.

\renewcommand{\baselinestretch}{1}
\begin{table}[hptb]
\centering
\caption{\small Relative  frequencies of  correct
estimation of  $r$  and $s$ by the Phillips-Perron test (PP.test) and
method (\ref{ACF}) in simulation
with 500 replications for Example 2.} \label{tab2}
\vskip3mm
{\small
\begin{tabular}{|c|c| c|c|c|c| c|c|}
\hline
 & $n$ &200 &300  &{500}& {1000} & {1500}& {2000}\\
 \cline{2-8}
  \raisebox{1ex}{(p, r, s)}& Method & r \; s & r \; s & r \;  s& r \; s & r \; s & r \;  s \\
 \hline
 \hline
  (6, 2, 2)& PP.test &.964 .412  &.970 .440     
  &.978 .420  &.982 .416  &.970 .448  &.960 .460 \\ 
\cline{2-8}
& (\ref{ACF})  &.614 .486  &.908 .766      
&.962 .814  &.944 .876  &.942 .892  &.924 .898   \\
 \hline
(6, 3, 1)& PP.test &.996 .288    &1.00 .336   
&.996 .342   &.992 .408   &.998 .416  &.998 .430   
\\
\cline{2-8}
 & (\ref{ACF})   &.904 .604    &.992 .782   
 &.998 .896   &.986 .924   &.992 .940  &.988 .958 \\ 
 \hline
(10, 4, 4)& PP.test &.840 .348   &.874 .392   
&.854 .392   &.852 .446    &.842 .430    &.824 .454 \\
\cline{2-8}
 & (\ref{ACF})  &.078 .162   &.538 .480   
 &.924 .798   &.940 .866    &.896 .858    &.880 .870 \\ 
 \hline
\cline{2-8}
  (10, 6, 2)& PP.test  &.984 .262   &.986 .276   
  &.978 .330   &.984 .322    &.978 .404   &.974 .406 
\\
\cline{2-8}
& (\ref{ACF})   &.566 .488   &.932 .740   
&.954 .826   &.942 .874    &.920 .876   &.910 .884 \\ 
 \hline
\cline{2-8}
 (15, 8, 4)&PP.test &.780 .192   &.792 .174  
 &.812 .218    & .750 .232         &.726 .260  &.658 .310 
\\
\cline{2-8}
& (\ref{ACF})  &.006 .110   &.326 .372  
&.868 .684     &.836 .708          &.858 .770  &.830 .768  \\ 
 \hline
\end{tabular}
 }
\end{table}

Also included in Table  \ref{tab2} are the results from
applying the Phillips-Perron unit-root test (PP.test), with  significance
level set at 0.01,  for estimating $r$; see (\ref{hyper0}).
By applying the same procedure to the differenced $\wh \bx_{t1}$, we also obtain the
estimated  $s$.
When $p$ is small, the PP.test estimates $r$ slightly better than
(\ref{ACF}) though both  methods perform well. For estimating $s$, the PP.test
is much worse than (\ref{ACF}).
When $p$ is large, (\ref{ACF}) performs substantially better than the
PP.test.
Also noticeable in   	
Table  \ref{tab2} is the fact that the larger  $r/p$ is, the more accurate are the estimates
for $r$, and the larger $s/(p-r)$ is,
the more accurate are the estimates for $s$.
Overall  (\ref{ACF}) provides more a stable performance than PP.test.

Figs~\ref{fig6}--\ref{fig7} present the boxplots of
$D_1(\calM(\wh \bA_2), \calM(\bB_2)) $ and $D_1(\calM(\wh \bA_{11}), \linebreak \calM(\bB_{11}))$
for $(p,r,s)=(6, 2, 2)$ and $(10, 4, 4)$ respectively, where
$\calM(\bB_{11})$ is the true cointegration space specified by the $I(1)$ components
of $\bx_{t1}$. As expected, the estimation errors decrease as sample size
$n$ increases, and the errors with $(p,r,s)=(10, 4, 4)$ are greater than
those with $(p,r,s)=(6, 2, 2)$.

\begin{figure}[htbb]
\begin{minipage}[t]{0.48\linewidth}
\centering
 \includegraphics[width=2.5in, height=2in]{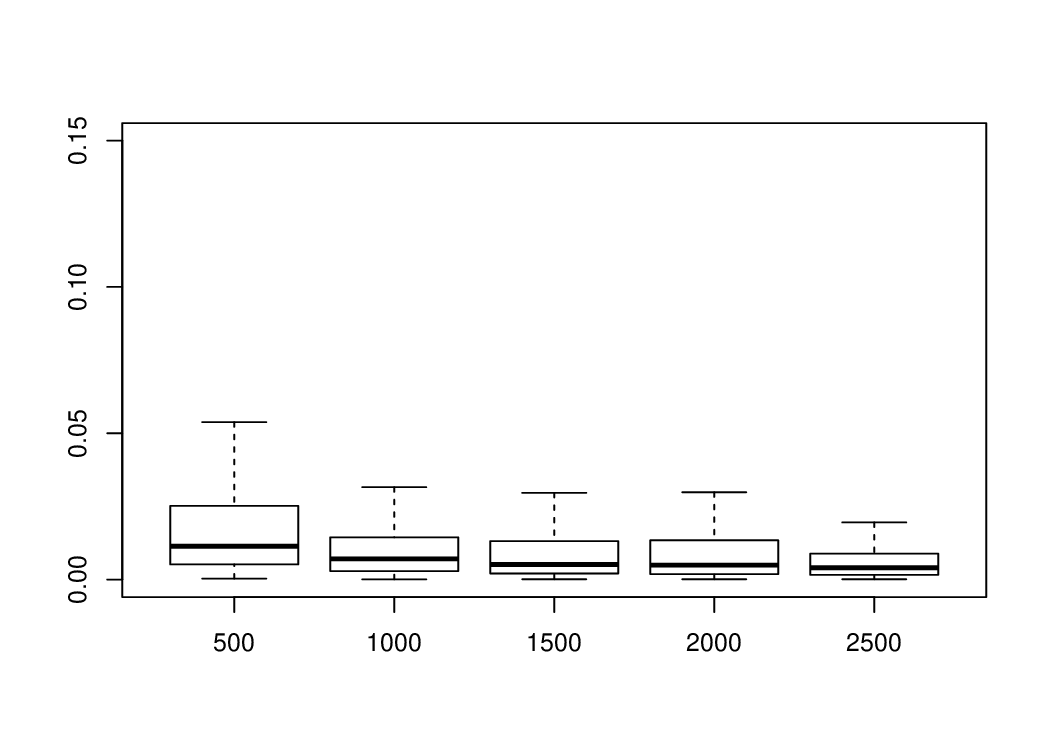}
\end{minipage}
\begin{minipage}[t]{0.48\linewidth}
\centering
 \includegraphics[width=2.5in, height=2.0in]{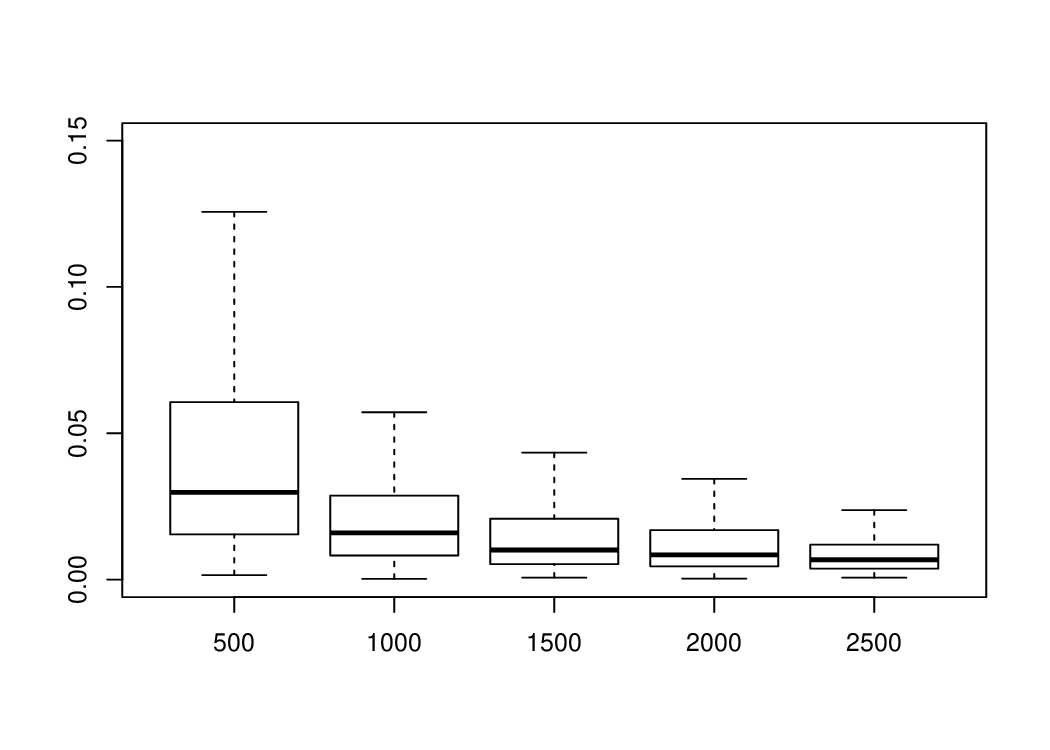}
\end{minipage}
\caption{Example 2: Boxplots of $D_1(\calM(\wh \bA_2), \calM(\bB_2))$
(left panel) and $D_1(\calM(\wh \bA_{11}), \calM(\bB_{11}))$ (right panel)
when $(p, r, s)=(6, 2,2)$. The labels on the horizontal axis are sample size $n$.}
\label{fig6}
\end{figure}

\begin{figure}[htbb]
\begin{minipage}[t]{0.48\linewidth}
\centering
 \includegraphics[width=2.5in, height=2in]{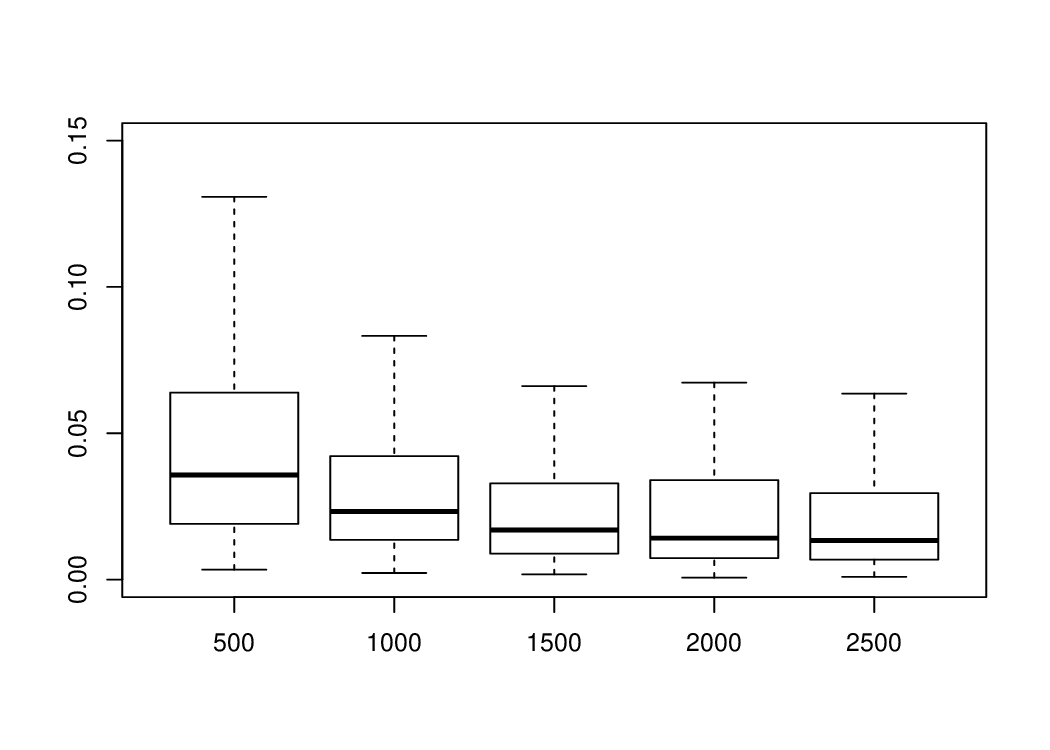}
\end{minipage}
\centering
\begin{minipage}[t]{0.48\linewidth}
\centering
 \includegraphics[width=2.5in, height=2in]{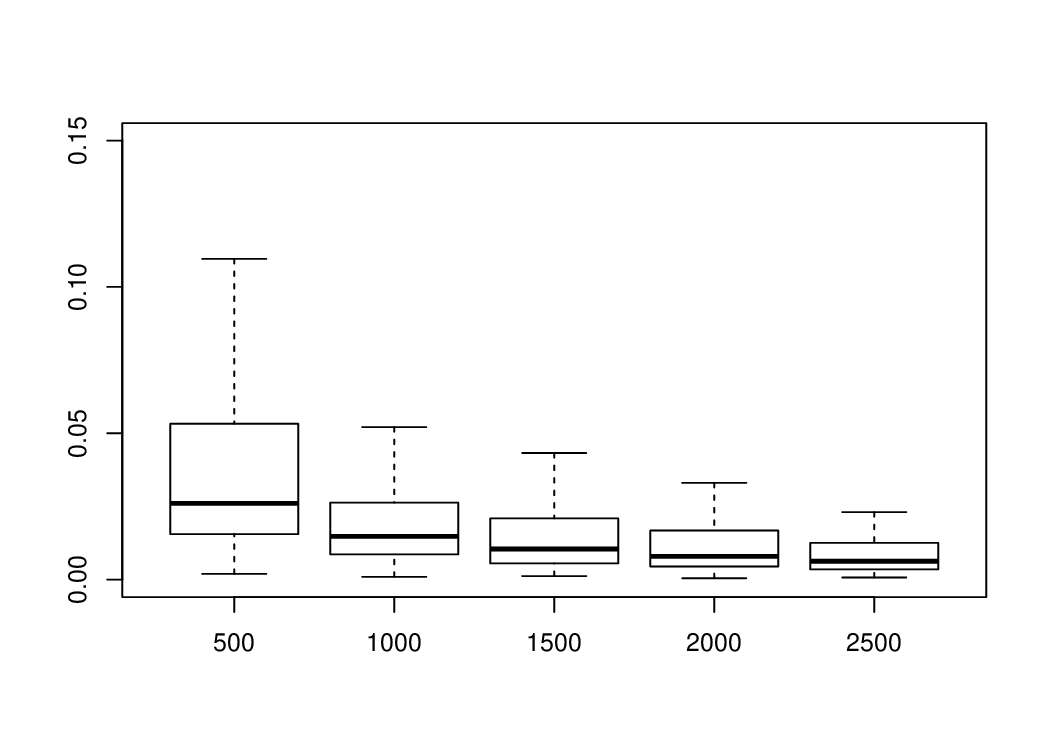}
\end{minipage}
\caption{Example 2: Boxplots of $D_1(\calM(\wh \bA_2), \calM(\bB_2))$
(left panel) and $D_1(\calM(\wh \bA_{11}), \calM(\bB_{11}))$ (right panel)
when $(p, r, s)=(10, 4,4)$. The labels on the horizontal axis are sample size $n$.}
\label{fig7}
\end{figure}

 {\bf Example 3.} Now we consider an example in which
the  components of $\by_{t}$ are $I(1)$ with linear trend, i.e.,
\beqn \label{trend}\by_{t}=\boldmu_1+\boldmu_2 t+\boldZ_t=\bA\bx^*_t\eeqn
for some $(\bx^*_t)'=(\bx^*_{t1},\bx_{t2})$, where
$\bx^*_{t1}=\boldmu^*_1+\boldmu^*_2 t+\bx_{t1}, \, \bx_{t1}$ is
nonstationary process and $\bx_{t2}$ is stationary process.
In our simulation, all component of $\boldmu^*_1$ and $\boldmu^*_2$  are
taken as  $0.3$ and $0.5$ respectively, all  components of
$\bx_{t2}$   are AR(1) with coefficients
generated  from $U(-0.8, 0.8)$,	
all components of $\bx_{t1}$ are ARIMA(1,1,1) with AR coefficients generated
from $U(0, 0.8)$ and MA  coefficients generated from $U(0, 0.95)$, and
all innovations are independent $N(0,1)$. Table \ref{tab3} reports the
relative frequencies of the occurrence of the event
$\{ \wh r=r\}$ and the average distance (\ref{4.1})
 in a simulation with 500 replications, where the cointegration rank is
estimated by  (\ref{ACF}) with $c_0=0.3$. Also included in Table
\ref{tab3} are the results obtained from
applying the Phillips-Perron unit-root test
to estimate $r$, see (\ref{hyper0}).  Table \ref{tab3} indicates that
(\ref{ACF}) works well even in the presence of a deterministic linear
trend, where our theoretical setting exclude.  However 
the Phillips-Perron test  performs poorly for large  $p$ and small $r/p$.

 \renewcommand{\baselinestretch}{1}
\begin{table}[hptb]
\centering
\caption{\small Relative  frequencies of  correct estimation of  $r$ and  average distance in simulation
with 500 replications for Example 3.} \label{tab3}
\vskip3mm
{\small
\begin{tabular}{|c|c|c|c|c|c|c|c|c|c|c|c|c|c|}
\hline
 &  &\multicolumn{2}{c|}{n=200} &\multicolumn{2}{c|}{n=300} &\multicolumn{2}{c|}{n=500}& \multicolumn{2}{c|}{n=1000} & \multicolumn{2}{c|}{n=1500}& \multicolumn{2}{c|}{n=2000} \\
 \cline{3-14}
 \raisebox{1ex}{(p,r)}& \raisebox{1ex}{Method}&RF& $D_1$ &RF& $D_1$ &RF& $D_1$ &RF& $D_1$ &RF& $D_1$ &RF& $D_1$ \\
 \hline
 & PP.test  &.882 &.087   &.780 &.143  
 &.664 &.200   &.950 &.030    &.838 &.096   &.746 &.150 \\  
 \cline{2-14}
 \raisebox{1ex}{(6, 2)}& New  &.452 &.331   &.858 &.107  
 &.982 &.026   &1.00 &.002    &1.00 &.004   &.998 &.008 \\ 
 \hline
 \hline
 & PP.test  &.988 &.010  &.988 &.007   
 &1.00 &.002   &.998 &.002   &.996 &.002  &.990 &.005 \\ 
 \cline{2-14}
\raisebox{1ex}{(6, 4)}& New   &.974 &.016  &1.00 &.002   
&1.00 &.002   &1.00 &.001   &1.00 &.002  &1.00 &5e-4   
\\
 \hline
 \hline
& PP.test &.842 &.092   &.398 &.293    
&.624  &.182    &.330 &.324   &.334 &.319   &.488  &.244 \\ 
 \cline{2-14}
\raisebox{1ex}{(10, 4)}& New   &.066 &.485   &.328 &.327   
&.966  &.042    &1.00 &.021   &1.00 &.007  &1.00   &.012  
 \\
 \hline
 \hline
 & PP.test  &.766 &.107  &.316  &.279   
 &.664 &.132    &.846 &.062   &.806 &.076   &.876 &.048 \\
 \cline{2-14}
\raisebox{1ex}{(10, 6)} & New  &.432 &.231  &.796  &.103   
&.998 &.010    &1.00 &.006   &1.00 &.003   &1.00 &.002  
\\
 \hline
 \hline
 & PP.test &.082 &.454   &.166 &.377   
 &.094 &.424    &.142 &.388    &.046  &.468  &.096 &.436\\
 \cline{2-14}
\raisebox{1ex}{(15, 6)}&New &0    &.651   &.004 &.521   
&.506 &.221    &.996 &.021    &.998  &.021  &1.00 &.005 
\\
\hline
\hline
& PP.test &.290 &.240   &.592 &.137  
&.336 &.217   &.484 &.157   &.798 &.064   &.446 &.177   
\\
 \cline{2-14}
\raisebox{1ex}{(15, 10)}& New  &.066 &.332   &.646 &.124  
&.964 &.034   &1.00 &.003   &1.00 &.004   &1.00 &.007  
\\
 \hline
 \hline
& PP.test &0 &.628     &0 &.667      
&0  &.671    &0   &.686   &0    &.696   &0    &.703    
\\
 \cline{2-14}
\raisebox{1ex}{(30, 10)}& New &0 &.769     &0 &.737      
&0  &.655   &.364 &.234   &.974 &.062   &.994 &.040    
\\
 \hline
 \hline
& PP.test &0 &.346   &.004 &.329    
&.010  &.324    &.010   &.314  &.034    &.294   &.006    &.329    
\\
 \cline{2-14}
\raisebox{1ex}{(30, 20)}& New &0 &.456   &.002 &.368    
&.344  &.168   &1.00 &.019   &1.00 &.010   &1.00 &.010    
\\
 \hline
\end{tabular}
}
\end{table}

{\bf Example 4.} We consider fractional cointegration cases now. Let the components of
$\bx_{t1}$ be $I(d)$ processes with a fractional order $d=4/5$ or 3/4,
the components of $\bx_{t2}$ be AR(1) with autoregressive coefficients $0.2i$
($i=1, \cdots, r$), the
elements of $\bA$ be generated independently from $U(-3, 3)$, and
all innovations be independent and $N(0,1)$.
We consider various combinations for $p, \, r$, $s$, and the sample size $n$.
For each setting, we replicate the simulation 500 times and
estimate the cointegration rank $r$ using (\ref{ACF}) with $c_0=0.3$.
The relative frequencies for the occurrence of the event $\{ \wh r =r\}$ and
the mean of distance (\ref{4.1}) over 500 replications are listed in Table~\ref{tab4}.
While the proposed methodology works well, the accuracy is slightly lower than that  integer cointegration orders. See the examples above. We also notice that
the estimation errors with $d=3/4$ are greater than those with $d=4/5$.

To illustrate the impact of the choice of $j_0$ on the estimation, we
consider the above fractional
cointegration with $p=6, \, r=4$
and order $d=4/5, \, 3/4$ and $2/3$. By setting  sample size $n=1000$ and $j_0$ between
5 and 100,
the relative frequencies for the occurrence of the event $\{ \wh r =r\}$ and
the mean of the distance (\ref{4.1}) are reported in Table~\ref{tab5}.
As mentioned in Section 2, using different values of $j_0$ hardly changes the results.

\renewcommand{\baselinestretch}{1}
\begin{table}[hptb]
\centering
\caption{\small
Relative  frequencies (RF) of the occurrence of event $\{\wh r = r\}$ and
 average distance $D_1$ defined in (\ref{4.1}) in simulation
with 500 replications for Example 4.} \label{tab4}
\vskip3mm
{
\small
\begin{tabular}{|c|c|c|c|c|c|c|c|c|c|c|c|c|c|}
\hline
 & & \multicolumn{2}{c|}{$n$=200}& \multicolumn{2}{c|}{$n$=300} 
 & \multicolumn{2}{c|}{$n$=500}& \multicolumn{2}{c|}{$n$=1000} & \multicolumn{2}{c|}{$n$=1500}& \multicolumn{2}{c|}{$n$=2000}\\
 \cline{3-14}
\raisebox{1ex}{$d$} & \raisebox{1ex}{$(p, \; r)$} & RF& $D_1$ &RF& $D_1$ &RF& $D_1$ & RF& $D_1$ &RF& $D_1$ &RF& $D_1$ \\
 \hline
&(3, 2) &.828 &.134  &.948 &.068  
& .978 &.040 &1.00 &.017 &.998 &.014 &1.00 &.010 \\
 \cline{2-14}
4/5& (6, 2)&.020 &.664  &.240 &.507  
&.664 &.294 &.946 &.119 &.966 &.101 &.986 &.070\\
  \cline{2-14}
&(9, 3)  &0 &.721   &.004 &.656   
&.188 &.488 &.766 &.250 &.868 &.181 &.920 &.156 \\
  \cline{2-14}
&(12, 4)   &0 &.743 &0 &0.701  
&.014 &.596 &.528 &.380 &.716 &.307 &.788 &.275\\
\hline
&(3, 2)  &.770 &.174   &.902 &.098   
&.964 &.058 &.984 &.033 &1.00 &.019 &.998 &.017 \\
\cline{2-14}
3/4& (6, 2) &.018 &.685   &.132 &.578    
&.488 &.380 &.866 &.193  &.916 &.151  &.942 &.118\\
\cline{2-14}
&(9, 3) &0 &.733   &0 &.680 
&.104 &.549 &.604 &.336 &.800 &.240 &.864 &.205\\
\cline{2-14}
&(12, 4) &0 &.754   &0 &.719   
&.006 &.629 &.328 &.450 &.606 &.378 &.696 &.344 \\
 \hline
\end{tabular}
}
\end{table}

\renewcommand{\baselinestretch}{1}
\begin{table}[hptb]
\centering
\caption{\small
 Relative  frequencies (RF) of the occurrence of event $\{\wh r = r\}$ and
 average distance $D_1$ defined in (\ref{4.1}) with $n=1000$ in simulation with $500$ replications for Example 4.} \label{tab5}
\vskip3mm
{\small 
\begin{tabular}{|c|c|c|c|c|c|c|c|c|c|c|c|c|}
\hline
 &\multicolumn{2}{c|}{$j_0$=5}& \multicolumn{2}{c|}{$j_0$=10} & \multicolumn{2}{c|}{$j_0$=15}& \multicolumn{2}{c|}{$j_0$=20}& \multicolumn{2}{c|}{$j_0$=50}&\multicolumn{2}{c|}{$j_0$=100}\\
 \cline{2-13}
\raisebox{1ex}{$d$} & RF& $D_1$ &RF& $D_1$ &RF& $D_1$ &RF& $D_1$&RF& $D_1$&RF& $D_1$\\
 \hline
4/5 &.964 &.086   &.984 &.069  &.982  &.062  &.982 &.064  &.982 &.054  &.980 &.057\\
 \hline
3/4 &.934 &.125  &.950 &.107   &.952  &.101  & .956 &.091   &.954 &.082  &.960 &.084  \\
  \hline
2/3&.788 &.226 &.788  &.209  &.788 &.199  &.804  &.195  &.814 &.171  &.806 &.179\\
 \hline
\end{tabular}
}
\end{table}

\noindent
{\bf Example 5.} We consider the 8 monthly US Industrial Production  indices for
January 1947 -- December 1993 published by the US Federal Reserve,
namely {\sl the total index, manufacturing
index, durable manufacturing, nondurable manufacturing, mining,
utilities, products} and {\sl materials}. The original 8 time series
are plotted in Fig.\ref{fig10}. Applying the proposed
method to these data, the transformed series $\wh \bx_t = \wh \bA' \by_t $ are plotted
in Fig.\ref{fig11} together with their sample ACF. The proposed method (\ref{ACF}) leads to
$\wh r =4$ with $m=40, \, c=0.3$ and $j_0=50$ or $100$.

\begin{figure}[htbb]
\centering
 \includegraphics[height=2.8in,width=5in]{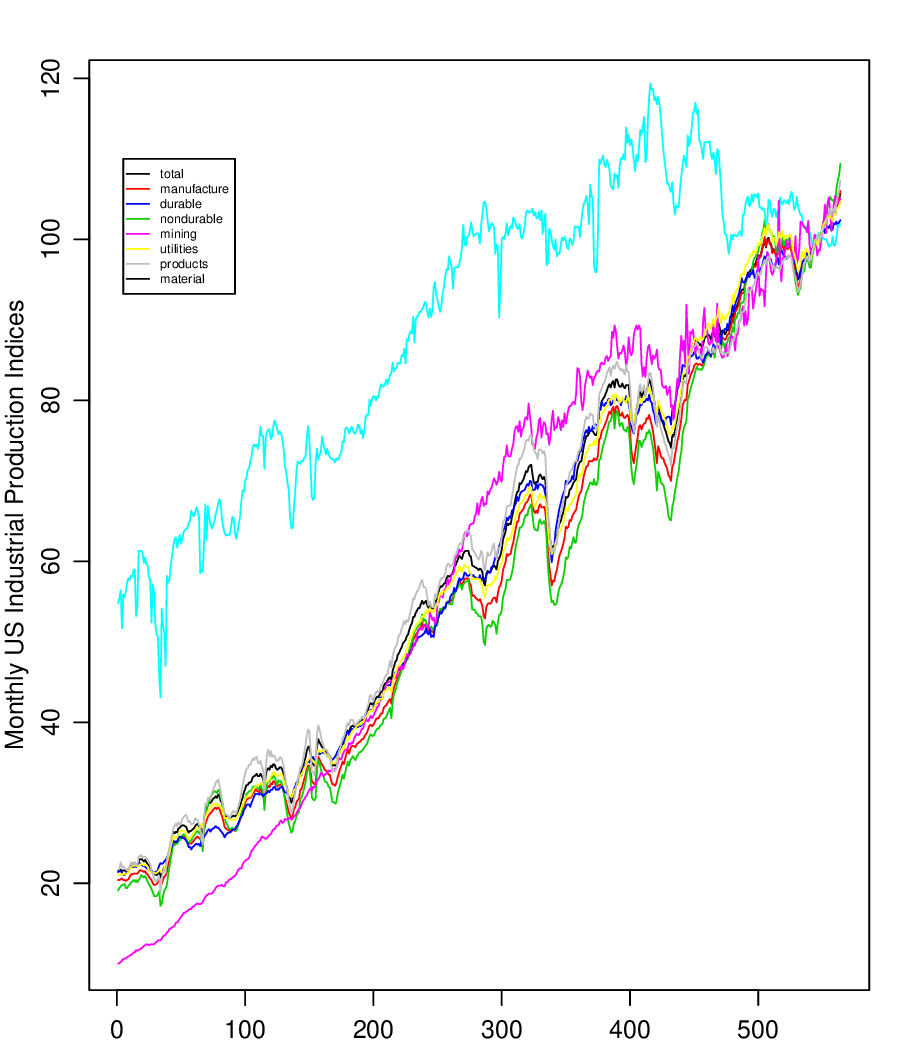}
\caption{Time series plots of the 8 monthly U.S. Industrial Production  indices
in January 1947 - December 1993.} \label{fig10}
\end{figure}

\begin{figure}[htbp]
\centering
\includegraphics[height=4in,width=5.5in]{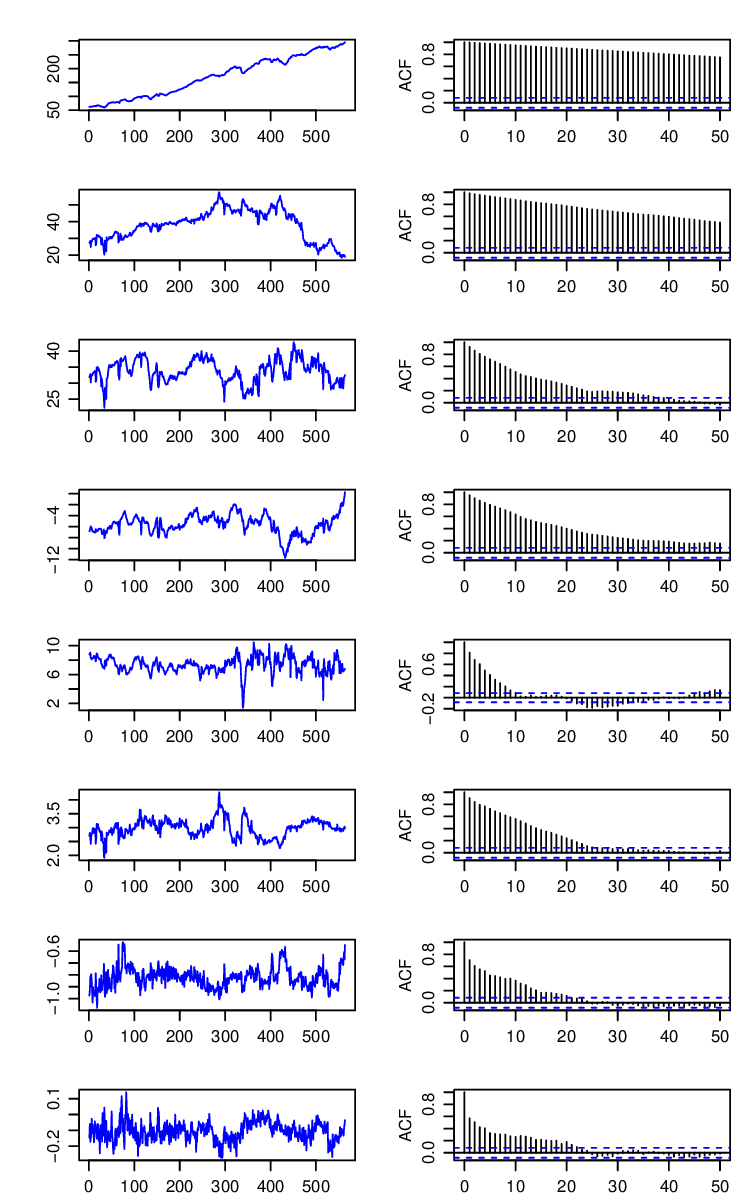}
\caption{Time series plots of the estimated $\wh \bx_t $ by the proposed method
and their sample ACF for the 8 monthly U.S. Industrial Production  indices.} \label{fig11}
\end{figure}

We also apply Johansen's (1991) likelihood method to this data set. Both the
trace and the maximum tests indicate $r=4$. The corresponding transformed
series together with their sample ACF are plotted in Fig.\ref{fig12}.

Let $\wh\bA_2$ denote the last 4 columns of $\wh \bA$ and $\wh\bB_2$ consist
of the loadings for the last 4 component series displayed in Fig.\ref{fig12},
i.e., the columns of $\wh\bA_2$  are the loadings of the 4
cointegrated variables identified by the proposed method in this paper, and
the columns of $\wh\bB_2$ are the loadings of the 4 cointegrated variables identified by
Johansen's likelihood method. Then
$$
D_1 ( \calM(\wh \bA_2), \calM ( \wh \bB_2) )^2 = 1 - {1\over 4} \tr\{\wh \bA_2 \wh \bA_2'
\wh\bB_2(\wh\bB_2'\wh\bB_2)^{-1} \wh\bB_2'\} = 1 - 0.9816 = 0.0184.
$$
This indicates that the two sets of cointegrated variables identified
by the two methods are effectively equivalent.

\begin{figure}[htbp]
\centering
\includegraphics[height=4in,width=5.5in]{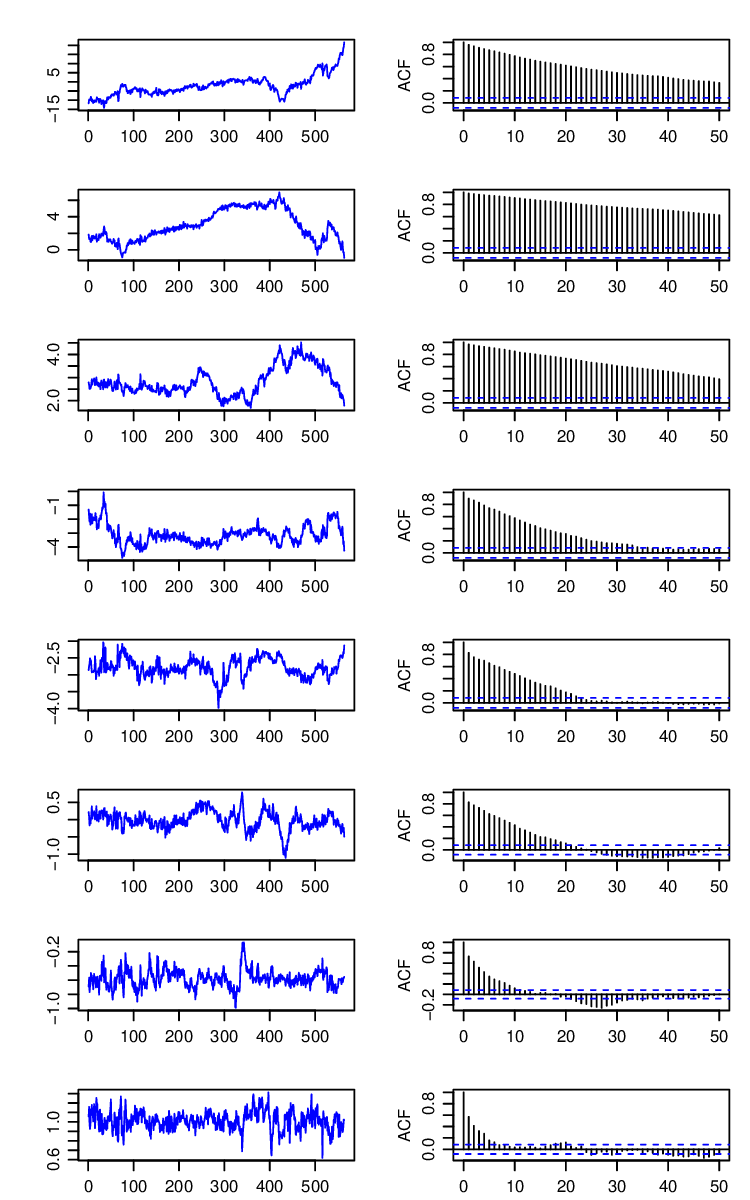}
\caption{Time series plots of the estimated $\wh \bx_t $ by Johansen's method
and their sample ACF for the 8 monthly U.S. Industrial Production  indices.} \label{fig12}
\end{figure}

To illustrate the impact of the choice of $c_0$ on the
estimation, we consider model (\ref{2.1}) with $p=2$ and the
following three specifications for $\bx_t$:
\begin{itemize}
   \item[](i) $r=0$, both  components of $\bx_{t}$ are ARIMA(1, 1, 1) processes with coefficient  $(0.6, 0.4)$ and $(0.8, 0)$,

    \item[] (ii) $r=1$, $x_{t1}$ is ARIMA(1, 1, 1) with $(0.6, 0.4)$ and
$x_{t2}$ is AR(1) with coefficient $0.6$,

     \item[] (iii) $r=2$,  $x_{t1}$ is AR(1)  with coefficient $0.6$ and
$x_{t2}$ is ARMA(1, 1) with coefficient $(0.6,  0.4)$.
    \end{itemize}
The
elements of $\bA$ are generated independently from $U(-3, 3)$ and $c_0$ is
taken from $0.05$ to $0.95$.
In each setting, we replicate the simulation 500 times with sample size
$n=200, 300, 500$ and $1000$. The relative frequencies for the occurrence
of the event $\{ \wh r =r\}$ are reported in Table 6.
When $r=0$,  smaller $c_0$ would lead to better performance,
however when $r=2$, larger $c_0$ may result in better performance. It is
because that as $r=0$, both the components are I(1),  smaller $c_0$ tends
to estimate $ r$ as $0$, while as $r=2$, both the components are
I(0), larger $c_0$ tends to estimate $ r$ as $2$, see Remark 1. Further,
it is shown that when $c_0$ is taken away from the endpoints, say $c_0\in
(0.2, 0.5)$, then the proposed procedure works well for all cases,
especially when $n$ is large.

Table 7 reports the simulation results with $p=3$, $\bA$ generated in
the same manner as the above, and three settings for $\bx_t$:
\begin{itemize}
   \item[](i) $r=0$, the components of $\bx_t$ are all ARIMA(1, 1, 1)
with coefficients
$(0.6, 0), \, (0.3, 0.7)$ and $(0.8, 0.4)$,
\item[](ii) $r=1$,
$x_{t1}$ and $x_{t2}$ are both ARIMA(1, 1, 1) with coefficients $(0.5, 0),
\, (0.8, 0.4)$, and $x_{t3}$ is AR(1) with coefficient $0.6$,
\item[](iii) $r=2$, $x_{t1}$ is ARIMA(1, 1, 1) with coefficient $(0.8, 0.4), \, x_{t2}$
is AR(1) with coefficient $0.6$ and $x_{t3}$ is ARMA(1, 1) with
coefficient $(0.5, 0.5)$.
\end{itemize}
The pattern of Table 7 is very similar to that of Table 6, i.e.
the estimation is stable for $c_0\in (0.2, 0.5)$.

\renewcommand{\baselinestretch}{1}
\begin{table}[hptb]
\centering
\caption{\small Relative  frequencies (RF) of the occurrence of event $\{\wh r = r\}$ for $p=2$ with different $c_0$ and 500 replications.} \label{tab6}
\vskip3mm
{\tiny
\begin{tabular}{|c|c|c|c|c|c|c|c|c|c|c|c|c|c|c|c|c|c|c|c|c|}
\hline
 r& n &.05 &.10 &.15 &.20 &.25 &.30 &.35 &.40 &.45 &.50 &.55 &.60 &.65 &.70 &.75 &.80 &.85 &.90 &.95\\
 \hline
 &200 &1.00 &1.00 &1.00 &.996 &.988 &.976 &.946 &.926 &.872 &.814 &.722 &.616 &.480 &.350 &.218 &.112 &.036 &.002 &0\\
 \cline{2-21}
 &300 &1.00 &1.00 &1.00 &1.00 &1.00 &1.00 &.998 &.996 &.984 &.958 &.932 &.880 &.796 &.672 &.522 &.364 &.180 &.034 &0\\
 \cline{2-21}
 0 &500 &1.00 &1.00 &1.00 &1.00 &1.00 &1.00 &1.00 &1.00 &1.00 &.998 &.994 &.990 &.980 &.946 &.868 &.754 &.546 &.232 &.016\\
 \cline{2-21}
 &1000 &1.00 &1.00 &1.00 &1.00 &1.00 &1.00 &1.00 &1.00 &1.00 &1.00 &1.00 &1.00 &1.00 &1.00 &.998 &.988 &.922 &.696 &.226\\
 \hline
 \hline
 &200 &.648 &.878 &.966 &.992 &.992 &.994 &.988 &.976 &.968 &.952 &.936 &.902 &.838 &.766 &.642 &.538 &.416 &.264 &.074\\
 \cline{2-21}
 &300 &.514 &.864 &.984 &.998 &1.00 &1.00 &1.00 &1.00 &.998 &.992 &.984 &.974 &.952 &.930 &.858 &.752 &.620 &.398 &.152\\
  \cline{2-21}
  1&500 &.442 &.882 &.990 &.998 &.998 &1.00 &.998 &.998 &.998 &.996 &.994 &.994 &.988 &.982 &.952 &.910 &.844 &.636 &.336\\
  \cline{2-21}
  &1000 &.210 &.900 &.992 &.998 &.998 &.998 &.998 &.998 &.998 &.998 &.998 &.996 &.996 &.996 &.994 &.990 &.980 &.910 &.634\\
  \hline
  \hline
  &200 &.230 &.616 &.866 &.970 &.994 &.998 &1.00 &1.00 &1.00 &1.00 &1.00 &1.00 &1.00 &1.00 &1.00 &1.00 &1.00 &1.00 &1.00\\
  \cline{2-21}
  &300 &.150 &.568 &.884 &.980 &.990 &1.00 &1.00 &1.00 &1.00 &1.00 &1.00 &1.00 &1.00 &1.00 &1.00 &1.00 &1.00 &1.00 &1.00\\
  \cline{2-21}
  2 &500 &.088 &.502 &.910 &.992 &1.00  &1.00 &1.00 &1.00 &1.00 &1.00 &1.00 &1.00 &1.00 &1.00 &1.00 &1.00 &1.00 &1.00 &1.00\\
  \cline{2-21}
  &1000 &.018 &.558 &.978 &1.00  &1.00  &1.00 &1.00 &1.00 &1.00 &1.00 &1.00 &1.00 &1.00 &1.00 &1.00 &1.00 &1.00 &1.00 &1.00\\
  \hline
\end{tabular}
}
\end{table}

\renewcommand{\baselinestretch}{1}
\begin{table}[hptb]
\centering
\caption{\small Relative  frequencies (RF) of the occurrence of event $\{\wh r = r\}$ for $p=3$ with different $c_0$ and 500 replications.} \label{tab7}
\vskip3mm
{\tiny
\begin{tabular}{|c|c|c|c|c|c|c|c|c|c|c|c|c|c|c|c|c|c|c|c|c|}
\hline
 r& n &.05 &.10 &.15 &.20 &.25 &.30 &.35 &.40 &.45 &.50 &.55 &.60 &.65 &.70 &.75 &.80 &.85 &.90 &.95\\
 \hline
 &200 &.986 &.982 &.962 &.938 &.896 &.826 &.744 &.608 &.490 &.370 &.276 &.178 &.096 &.036 &.008 &.002 &0 &0 &0
\\
 \cline{2-21}
 &300 &1.00 &1.00 &.998 &.998 &.994 &.982 &.962 &.926 &.866 &.806 &.702 &.562 &.390 &.238 &.118 &.040 &.006 &0 &0\\
 \cline{2-21}
 0 &500 &1.00 &1.00 &1.00 &1.00 &1.00 &1.00 &1.00 &.998 &.994 &.986 &.972 &.940 &.864 &.738 &.526 &.316 &.112 &.008 &0\\
 \cline{2-21}
 &1000 &1.00 &1.00 &1.00 &1.00 &1.00 &1.00 &1.00 &1.00 &1.00 &1.00 &1.00 &1.00 &.998 &.994 &.978 &.876 &.692 &.282 &.004\\
 \hline
 \hline
 &200 &.732 &.932 &.958 &.968 &.954 &.932 &.908 &.860 &.788 &.718 &.632 &.530 &.430 &.308 &.188 &.096 &.032 &0 &0\\
 \cline{2-21}
 &300 &.584 &.900 &.986 &1.00 &1.00 &.996 &.988 &.974 &.956 &.924 &.880 &.816 &.724 &.600 &.418 &.268 &.114 &.022 &0\\
  \cline{2-21}
  1&500 &.456 &.896 &.984 &.994 &.998 &1.00 &1.00 &.994 &.990 &.988 &.988 &.972 &.944 &.908 &.824 &.676 &.446 &.198 &.032\\
  \cline{2-21}
  &1000 &.258 &.900 &.996 &.996 &.998 &.998 &1.00 &.998 &.998 &.996 &.996 &.996 &.994 &.990 &.990 &.962 &.884 &.626 &.194\\
  \hline
  \hline
  &200 &.288 &.780 &.964 &.998 &1.00 &1.00 &1.00 &1.00 &1.00 &1.00 &.998 &.990 &.962 &.942 &.886 &.828 &.724 &.522 &.252\\
  \cline{2-21}
  &300 &.448 &.814 &.944 &.990 &.998 &.998 &.998 &.994 &.992 &.982 &.962 &.934 &.878 &.820 &.756 &.666 &.500 &.322 &.126\\
  \cline{2-21}
  2 &500 &.210 &.786 &.982 &1.00 &1.00 &1.00 &1.00 &1.00 &1.00 &1.00 &1.00 &1.00 &1.00 &.998 &.978 &.950 &.892 &.726 &.420\\
  \cline{2-21}
  &1000 &.096 &.848 &.996 &1.00  &1.00  &1.00 &1.00 &1.00 &1.00 &1.00 &1.00 &1.00 &1.00 &1.00 &1.00 &1.00 &.996 &.960 &.714\\
  \hline
\end{tabular}
}
\end{table}

\section{Conclusions}
\setcounter{equation}{0}

We propose in this paper a simple,  direct and model-free method  for
identifying cointegration relationships among multiple
time series of which different components series may have different
integration orders. The method boils down to an eigenanalysis
for a non-negative definite matrix.
One may view that the components of the transformed series $\wh \bx_t =
\wh \bA' \by_t$ are arranged in  ascending order according to the ``degree''
of stationarity; reflected by the magnitude of the eigenvalues of $\wh \bW$.

\section{Appendix: Technical proofs}
\label{sec7}
\setcounter{equation}{0}

\subsection{Proof for Section \ref{sec31}}
Let
\beqn \label{6.14}\bSigma^{x}_j&=& \mathrm{diag}\left[\left({1\over n}\sum_{t=1}^{n-j}(\bold{x}_{t+j, 1}-\overline{\bold{x}}_{1})
(\bold{x}_{t1}-\overline{\bold{x}}_{1})'\right), \left({1\over n}\sum_{t=1}^{n-j} (\bold{x}_{t+j, 2}-\overline{\bold{x}}_{2})
(\bold{x}_{t2}-\overline{\bold{x}}_{2})'\right)\right]\nn\\
&\equiv &\mathrm{diag}(\bSigma_{j1}^x, \bSigma_{j2}^x),
\nn\eeqn
$\bW^{x}=\sum_{j=0}^{j_0}\bSigma^{x}_j(\bSigma^{x}_j)'=:\mathrm{diag}(\bD_1^x, \bD_2^x)$ and
 $\bGamma_x$ be the $p\times p$ orthogonal matrix such that
  \beqn \bW^{x}\bGamma_x=\bGamma_x \bLambda_x,\nn\eeqn
  where $\bLambda_x$ is the diagonal matrix of eigenvalues of $\bW^{x}$.
Since  $\bold{x}_{t1}$ is  nonstationary  and $\bold{x}_{t2}$ is  stationary,  intuitively ${1\over n}\sum_{t=1}^{n-j}(\bold{x}_{t+j, 1}-\overline{\bold{x}}_{1})$
$(\bold{x}_{t1}-\overline{\bold{x}}_{1})'$ and ${1\over n}\sum_{t=1}^{n-j}(\bold{x}_{t+j,2}-\overline{\bold{x}}_{2})
(\bold{x}_{t2}-\overline{\bold{x}}_{2})'$ do not share the same  eigenvalues, so $\bGamma_x$ must be  block-diagonal. Define $\bW^{y}=\bA \bW^{x}\bA'$, then
\beqn \bW^{y}=\bA \bW^{x}\bA'=\bA\bGamma_x\bLambda_x \bGamma_x' \bA'.\nn\eeqn
This implies that the columns of $\bA\bGamma_x$ are just the orthogonal eigenvectors of  $\bW^{y}.$ Since $\bGamma_x$ is  block-diagonal, it follows that
 $\mathcal{M}(\bA_2)$ is the same as the space spanned by  the eigenvectors corresponding to the smallest $r$ eigenvalues of $\bW^{y}.$ As a result,
to show the distance between the cointegration space and its estimate is small,
we only need to show that  the space spanned by the eigenvectors of $\bW^{y}$
can be approximated by that of $\hat{\bW}$. This question is usually solved by  perturbation matrix theory. In particular, let
\beqn \hat{\bW}=\bW^{y}+\Delta \bW^{y}, \, \, \,  \Delta \bW^{y}=\hat{\bW}-\bW^{y},\nn\eeqn
and
\beqn \mathrm{sep}(\bD_1^x, \bD_2^x)=\min_{\lambda\in \lambda(\bD_1^x), \, \mu\in \lambda(\bD_2^x)}|\lambda-\mu|,\nn\eeqn
where $\lambda(\bA)$ denotes the  set of eigenvalues of a matrix $\bA.$
When $||\Delta \bW^{y}||=o_p(\mathrm{sep}(\bD_1^x, \bD_2^x))$,  one can use
 the perturbation results of Golub and Loan (1996) to establish the bound of Theorems 1, 3 and 5, see also Lam and Yao (2012) or Chang, Guo and Yao (2017). However, in our setting $\mathrm{sep}(\bD_1^x, \bD_2^x)$ can be of smaller order than $||\Delta \bW^{y}||$, i.e., $\mathrm{sep}(\bD_1^x, \bD_2^x)/||\Delta \bW^{y}||\stackrel{p}{\longrightarrow} 0$ as $n\rightarrow\infty$ and the
above method will not work.

 To fix this problem, we  adopt the perturbation results of Dopico, Moro
and Molera (2000) instead. A similar idea was  used by Chen and Hurvich
(2006) to recover their fractional cointegration spaces via the
periodogram matrix, using a   random diagonal  block matrix instead.
However, because of the quadratic form of $\bW^x$
($=\sum_{j=1}^{j_0}\bSigma_j^x(\bSigma_j^x)'$), we cannot find a
normalizing constant matrix  $\bC_n$ such that $\bC_n\bW^x\bC_n=O_e(1)$ or
$\bC_n\bW^y\bC_n=O_e(1)$, so as a result, the argument of Chen and Hurvich (2006)
based on the perturbation bound of Barlow and Slapnicar (2002) cannot be used.
To this end, we first establish some lemmas (i.e. Lemmas 7-10  below) and we
legate their proofs to  supplementary material.

For $1\leq i\leq p-r$, set
 $f_{0}^{i}(t)=W^{i}(t), \, f_{d_i}^{i}(t)=\int_{0}^{t}f_{d_i-1}^{i}(s)\, dt, \, \mu_i=\mathrm{E}z_t^i$ and define
\beqn F^i(t)&=&f_{d_i}^i(t)-\int_{0}^{1}f_{d_i}^i(t)\, dt, \quad G_{d}(t)={\prod_{j=0}^{d-1}(t+j)\over d !}, \quad \bar{G}_{d}={1\over n}\sum_{t=1}^{n}G_{d}(t).\nn\eeqn
  Then, we have the following weak convergence result for the sample autocovariance.

  \begin{lemma} Let $L_{d}(t)=G_{d}(t)-\bar{G}_{d}$.
  Suppose $x_t^{i}\sim I(d_i), \, 1\leq i\leq p-r$, then under  Condition 1,
  \begin{equation}  \label{l7a} \Big({{x_t^i-\bar{x}^i-\mu_i L_{d_i}(t)}\over n^{d_i-1/2}}, \, 1\leq i\leq p-r\Big)\stackrel{d}{\longrightarrow}\Big(F^{i}(t), \, 1\leq i\leq p-r\Big) \, \, \hbox{and}\end{equation}
\begin{equation} \label{l7b}  \Big({1\over n^{d_i+1/2}}\sum_{t=1}^{n}(x_t^i-\bar{x}^i-\mu_iL_{d_i}(t))(x_t^j-\mathrm{E}x_t^j), i\leq p-r, p-r+1\leq j\leq p\Big)\stackrel{p}{\longrightarrow}
\bold{0}. \end{equation}

\end{lemma}

Next, we  establish a bound for the eigenvalues of $\bSigma_j^x$
and $\bA'\hat{\bSigma}_j\bA=:\hat{\bSigma}_j^{x}$.

Without loss of generality, we assume the first $r_1$ components of $\bold{x}_{t1}$ are $I(a_1)$, the next $r_2$ components  are $I(a_2)$  and the last $r_q$ components  of $\bold{x}_{t1}$ are $I(a_q)$, that is,
\beqn\boldx_{t1}=( \stackrel{I(a_q)}{\overbrace{x_{t}^{1}, \cdots, x_t^{r_q}}}, \, \, \, \stackrel{I(a_{q-1})}{\overbrace{x_{t}^{r_q+1}, \cdots, x_{t}^{r_q+r_{q-1}}}}, \, \cdots, \stackrel{I(a_1)}{\overbrace{x_{t}^{\sum_{j=2}^{q}r_{j}+1}, \cdots, x_t^{\sum_{j=1}^{q}r_{j}}}} )',\nn\eeqn where $ a_1<a_2<\cdots<a_q$ are positive integers and $ \sum_{i=1}^{q}r_i=p-r.$
 For  $ 1\leq i\leq q$,  define
  $\nu_q=0$ and $\nu_i=\sum_{j=i+1}^{q}r_{j}$. Then for any
$\bx_t(r_i):=(x_{t}^{\nu_i+1}, \cdots, x_{t}^{\nu_i+r_i})',$  if $\bmu_i:=(\mu_{\nu_i+1}, \cdots, \mu_{\nu_i+r_i})'\neq 0$,
there must exist a $r_i\times (r_i-1)$ matrix $\bP_i$ and  $r_i\times 1$ vector $\bar{\bmu}_i$  such that
$\bP'_i\bP_i=\bold{I}_{(r_{i}-1)}, \, (\bP_i, \bmu_i)$ has full rank $r_i,$
$\, \bP'_i\bmu_i=0$ and
 $\bar{\bmu}'_i\bmu_i=1$, where $\bold{I}_{a}$ denotes $a\times a$ matrix. Let
$\bB_i=(\bP_i, n^{-1/2}\bar{\bmu}_i)'$ if $\bmu_i\neq 0$ and $\bB_i=\bold{I}_{r_i}$ if $\bmu_i= 0,$  and $\bTheta_n=\mathrm{diag}(\bB_q,  \cdots, \bB_2, \bB_1, \bold{I}_{r}).$
Define
\beqn
\bD_{n1}=\mathrm{diag}\Big( \stackrel{r_q}{\overbrace{n^{a_q-1/2}, \cdots, n^{a_q-1/2}}}, \cdots,  \stackrel{r_1}{\overbrace{n^{a_1-1/2}, \cdots, n^{a_1-1/2}}}\Big), \,  \, \bD_{n2}=(\stackrel{r}{\overbrace{1, \cdots, 1}}),  \nn\eeqn
and $\bD_n=:\mathrm{diag}(\bD_{n1}, \bD_{n2}).$ Let $H^{d}(t)= {t^{d}/ d!}- {1/ (d+1)!}$, $F^{i}(t)$ be given as in Lemma 7,
 $\bF_i(t)=(F^{\nu_i+1}(t),  \cdots, F^{\nu_i+r_i}(t))', \,$
  $\bold{M}_i(t)=(\bF'_i(t)\bP_i, H^{a_i}(t))'I(\bmu_i \neq 0)+ \bF_i(t)I(\bmu_i = 0),$ and
$\bold{M}(t)=(\bold{M}'_q(t),  \bold{M}'_{q-1}(t), \cdots,
\bold{M}'_1(t))'.$ Then Lemma 8 below follows from Lemma 7 and the
continuous mapping theorem.

\begin{lemma} Let $\bGamma_j(x)=\hbox{diag}\Big({1\over n}\sum_{t=1}^{n}(\bx_{t1}-\bar{\bx}_1)(\bx_{t1}-\bar{\bx}_1)', \, \mathrm{Cov}(\bx_{1+j,2}, \, \bx_{1,2})\Big).$ Under Condition 1, we have
\beqn \label{6.23} \bD_n^{-1}\bTheta_n\bGamma_j^{x}\bTheta'_n\bD_n^{-1}
\stackrel{d}{\longrightarrow}\mathrm{diag}\Big(\int_{0}^{1}\bold{M}(t)\bold{M}'(t)\, dt, \, \mathrm{Cov}(\bx_{1+j,2}, \, \bx_{1,2})\Big).\nn\eeqn
\end{lemma}

 Let $F^{i}(t), \, 1\leq i \leq p-r$ be defined in Lemma 7, where
$W^{i}(t)=\sigma_{ii}B^i(t)$ and  $B^{i}(t), \,  1\leq i\leq p-r$ are
independent Brownian motions. Let $\bold{F}(t)=(F^{1}(t), F^{2}(t),
\cdots, F^{p-r}(t))'$.

\begin{lemma} Under condition 2 and $p=o(n^{1/2-\tau})$  with $ 0<\tau<1/2,$
\begin{equation}\label{l9a} \Big{\|}\bD_{n}^{-1}\bGamma_j^{x}\bD_{n}^{-1}-\mathrm{diag}\Big(\int_{0}^{1}\bold{F}(t)\bold{F}'(t)\, dt, \, \mathrm{Cov}(\bx_{1+j,2}, \, \bx_{1,2})\Big)\Big{\|}_2=o_p(1).\end{equation}
Further, $\int_{0}^{1}\bold{F}(t)\bold{F}'(t)\, dt$ is positive definite.
\end{lemma}

\begin{lemma}
Under Condition 1, or Condition 2 and $p=o(n^{1/2-\tau})$, we have
\beqn \label{l6.3a} \max_{0\leq j\leq j_0}\|\bD_n^{-1}\bTheta_n(\bSigma_j^{x}-\bGamma_j^{x})\bTheta'_n\bD_n^{-1}\|_2
\stackrel{p}{\longrightarrow} 0 \, \, \hbox{and}\eeqn
\beqn \label{l6.3b} \max_{0\leq j\leq j_0}\|\bD_n^{-1}\bTheta_n(\hat{\bSigma}_j^{x}-\bGamma_j^{x})\bTheta'_n\bD_n^{-1}\|_2
\stackrel{p}{\longrightarrow} 0.\eeqn
\end{lemma}

\noindent{\bf Proof of Theorem 1.}\quad Since
\beqn \{D(\hat{\mathcal{M}}(\bA_2), \mathcal{M}(\bA_2))\}^2&=&{1\over r}\{\mathrm{tr}[\bA'_2(I_p-\hat{\bA}_2\hat{\bA}_2^{'})\bA_2]\}\nn\\
&\leq &||\bA'_2(\bA_2\bA'_2-\hat{\bA}_2\hat{\bA}_2^{'})\bA_2||_2\leq 2 ||\hat{\bA}_2- \bA_2||_2^2,\nn\eeqn
it follows from Theorem I.5.5 of Stewart and Sun (1990) (see also Proposition 2.1 of Vu and Lei (2013)) that
\begin{equation} \label{6.27}D(\hat{\mathcal{M}}(\bA_2), \mathcal{M}(\bA_2))\leq \sqrt{2}||\hat{\bA}_2- \bA_2||_2\leq \sqrt{2}||\hat{\bA}_2- \bA_2||_F\leq 2\sqrt{2}||\sin\Theta(\hat{\bA}_2, \bA_2)||_F,\end{equation}
where $\Theta(\hat{\bA}_2, \bA_2)=\arccos[(\bA'_2\hat{\bA}_2\hat{\bA}_2\bA_2)^{1/2}]$ is the canonical angle between the column spaces of $\hat{\bA}_2$ and $\bA_2$. Let $\eta=\min_{\lambda\in \lambda(\bD_1^x), \, \mu\in \lambda(\tilde{\bD}_2^x)}|\lambda-\mu|/\sqrt{\lambda\mu},$ where $\lambda(\tilde{\bD}_2^x)$ consists of  the $r$ smallest   eigenvalues of $\bA'\hat{\bW}\bA=:\hat{\bW}^x.$
By Theorem 2.4 of Dopico, Moro and Molera (2000), we have
\beqn \label{6.28}||\sin\Theta(\hat{\bA}_2, \bA_2)||_F\leq ||(\bW^{y})^{-1/2}\Delta \bW^y (\hat{\bW})^{-1/2}||_F/\eta.\eeqn
Note that
\begin{equation} \label{6.29}(\bW^{y})^{-1/2}\Delta \bW^y (\hat{\bW})^{-1/2}
=
(\bW^{y})^{-1/2}(\hat{\bW})^{1/2}-(\bW^{y})^{1/2}(\hat{\bW})^{-1/2}.\end{equation}
Thus, by equations (\ref{6.27}), (\ref{6.28}) and (\ref{6.29}), we have
\beqn \label{6.30} D(\hat{\mathcal{M}}(\bA_2), \mathcal{M}(\bA_2))\leq
(||(\bW^{y})^{-1/2}(\hat{\bW})^{1/2}||_F+||(\bW^{y})^{1/2}(\hat{\bW})^{-1/2}||_F)/\eta.\nn\eeqn

Next, we show that$||(\bW^{y})^{-1/2}(\hat{\bW})^{1/2}||_F=O_p(1),$
 which is equivalent  to
 \beqn \label{6.30a} ||(\bW^{x})^{-1/2}(\hat{\bW}^{x})^{1/2}||_F=O_p(1).\eeqn
Note that
\begin{equation} \label{6.31} 0\leq \bSigma_0^{x}\leq (\bW^{x})^{1/2}\leq \sum_{j=0}^{j_0} \{\bSigma_j^{x}(\bSigma_j^{x})'\}^{1/2} \, \, \, \hbox{and} \, \, \, 0\leq \hat{\bSigma}_0^{x}\leq (\hat{\bW}^{x})^{1/2}\leq \sum_{j=0}^{j_0} \{\hat{\bSigma}_j^{x}(\hat{\bSigma}_j^{x})'\}^{1/2}.\end{equation}
It follows from (\ref{6.31})  that
\beqn ||(\bW^{x})^{-1/2}(\hat{\bW}^{x})^{1/2}||_F\leq \sum_{j=0}^{j_0} ||(\bSigma_0^{x})^{-1} \{\hat{\bSigma}_j^{x}(\hat{\bSigma}_j^{x})'\}^{1/2}||_F.\nn\eeqn
Thus, for (\ref{6.30a}), it is enough to show the eigenvalues of $(\bSigma_0^{x})^{-1}\sum_{j=0}^{j_0} \{\hat{\bSigma}_j^{x}(\hat{\bSigma}_j^{x})'\}^{1/2}$ are $O_p(1)$, which is equivalent to
\beqn\label{6.34} \hbox{the solutions} \, \,  \lambda \, \,\hbox{of} \, \,  | \{\hat{\bSigma}_j^{x}(\hat{\bSigma}_j^{x})'\}^{1/2}-\lambda \bSigma_0^{x}|=0 \, \, \hbox{are} \, \,O_p(1).\eeqn

 Since $\mathrm{diag}\left(\int_{0}^{1}\bold{M}(t)\bold{M}'(t)\, dt, \, \mathrm{Var}(\bold{x}_{1,2})\right)>0$, by  Lemma 10  the  solutions ($\lambda$) of equation
\beqn \label{6.35}|\bD_n^{-1}\bTheta_n\{\hat{\bSigma}_j^{x}(\hat{\bSigma}_j^{x})'\}^{1/2}
\bTheta'_n\bD_n^{-1}-\lambda \bD_n^{-1}\bTheta_n\bSigma_0^{x}\bTheta'_n\bD_n^{-1}|=0\eeqn
are bounded in probability. Thus, we have (\ref{6.34}) and  (\ref{6.30a}) as desired.

Similarly, we can show
\beqn\label{6.37}||(\bW^{y})^{1/2}(\hat{\bW})^{-1/2}||_F=
||(\bW^{x})^{1/2}(\hat{\bW}^{x})^{-1/2}||_F=O_p(1).\eeqn
 Using equations (\ref{6.31}) and (\ref{6.37}), the remainder of the proof of  Theorem 1 consists of showing that there exist two positive constants $c_1, c_2$ such that in probability $\eta\geq c_1n^{2a_1-1}/\sqrt{j_0}$ provided
$|I_0|\geq 2$ or $ |I_0|=1$ and
$\mathrm{E}z_t^{I_0}= 0$ and $\eta\geq c_2n^{2a_1}/\sqrt{j_0} $ provided $ |I_0|=1$ and
$\mathrm{E}z_t^{I_0}\neq 0.$

 Define $\lambda_i(\bA)$ to be the $i$-th eigenvalue of a matrix $\bA$.
  Note that
  $$\mathrm{diag}\left(\int_{0}^{1}\bold{M}(t)\bold{M}'(t)\, dt, \, \mathrm{Var}(\bold{x}_{1,2})\right)>0.$$
By Lemmas 8 and 10, it follows  that  when $|I_0|\geq 2$ or $ |I_0|=1$ and
$\mathrm{E}z_t^{I_0}= 0$,
$\lambda_{p-r}(\bSigma_j^x)=O_{e}(n^{2a_1-1})$ and $\lambda_{p-r+1}(\hat{\bSigma}_j^x)=O_{e}(1)$. Thus,
there exist two positive constants $c_3, c_4$ such that in probability
\begin{equation} \label{6.38}\lambda_{p-r}(\bW^x)\geq
\lambda_{p-r}(\bSigma_0^x(\bSigma_0^x)')\geq c_3n^{2(2a_1-1)}\end{equation}
and
\begin{equation}\label{6.39} c_3\leq\lambda_{p-r+1}(\hat{\bSigma}_0^x(\hat{\bSigma}_0^x)')\leq\lambda_{p-r+1}(\hat{\bW}^x)
\leq \Big[\lambda_{p-r+1}\Big(\sum_{j=0}^{j_0} \{\hat{\bSigma}_j^{x}(\hat{\bSigma}_j^{x})'\}^{1/2}\Big)\Big]^2\leq c_4j_0^2.\end{equation}
Hence,  in probability
 \beqn\label{6.40} \eta\geq |c_3n^{2(2a_1-1)}-c_4 j_0^2|/\sqrt{c_3n^{2(2a_1-1)}c_4 j_0^2}\geq c'n^{2a_1-1}/j_0.\nn\eeqn
Similarly, we have  $ |I_0|=1$ and
$\mathrm{E}z_t^{I_0}\neq 0$, then in probability,
\beqn \label{6.40a} \eta\geq  c'n^{2a_1}/j_0.\eeqn
Since $j_0$ is fixed, combining  (\ref{6.30a}), (\ref{6.40}) and (\ref{6.40a}), we complete the proof of (i) and (ii). Conclusion (iii) can be shown similarly by treating $A_{1i}$ as the role of $A_2$, see also the proof of Theorem 1 of Chen and Hurvich (2006), we omit the details here.
\hfill$\Box$

\vskip3mm

Let $\bA_{1,0}=\bA_2$ and $\wh \bB_{1i}= (\wh\gamma_{\nu_i+1}, \cdots, \wh\gamma_{\nu_i+r_i})$ for $i=1, \cdots, q$ and
$\wh\bB_{10}=(\wh\gamma_{p-r+1}, \cdots, \wh\gamma_{p})$.

\begin{lemma}\label{lem11} Under Condition 1, we have
\beqn \|\bB_{1,l}\bA_{1,h}\|_{F}=O_p(n^{-2|a_h-a_l|}), \, \, \hbox{for} \, \, l\neq h.\nn\eeqn
\end{lemma}

\begin{proof}  Let $\eta(\bB_{1,l}, \bA_{1,h})$ be defined as $\eta$  above, i.e.,
\beqn  \eta(\bB_{1, l}, \bA_{1,h})=\min_{\lambda\in \{\wh\lambda_{\nu_l+1}, \cdots, \wh\gamma_{\nu_l+r_l}\}, \, \mu\in \{\lambda_{\nu_h+1}, \cdots, \lambda_{\nu_h+r_h}\}}|\lambda-\mu|/\sqrt{\lambda\mu}.\nn\eeqn
By Lemmas 8 and 10, using the same arguments as  in Theorem 1, we have
\beqn \eta(\bB_{1,l}, \bA_{1,h})\geq cn^{2|a_h-a_l|}\eeqn
for some $c>0$.
It has been shown in Theorem 1 that $||(\bW^{y})^{-1/2}\Delta \bW^y (\hat{\bW})^{-1/2}||_F=O_p(1)$,  thus by
Theorem 2.4 of Dopico, Moro and Molera (2000) (see also Theorem 4.1 of Barlow and Slapni\u{c}ar (2000)), we have
\beqn \|\bB_{1,l}\bA_{1,h}\|_{F}&\leq& ||(\bW^{y})^{-1/2}\Delta \bW^y (\hat{\bW})^{-1/2}||_F/\eta(\bB_l, \bA_h)\nn\\
&=&O_p(n^{-2|a_h-a_l|}).\nn\eeqn
This completes the proof of Lemma\ref{lem11}.
\end{proof}

\vskip3mm

\noindent{\bf Proof of Theorem 2.}\quad First, we prove the consistency of $\hat{r}$. For any
$1\leq i\leq p$,
\beqn\label{7.18} \wh x_{t}^{i}=\wh\gamma'_i\by_t
=(\wh\gamma'_i\bA_{1q}\bx_{t1q}, \cdots, \wh\gamma'_i\bA_{11}\bx_{t11}, \wh\gamma'_i\bA_{2}\bx_{t2}).\eeqn
Let
$\nu_i$ be defined as in Lemma 7 and $r_0=r.$ By Lemma 11,
when $\nu_l+1\leq i\leq \nu_l+r_l, \, l\neq h$,
\beqn \wh\gamma'_i\bA_{1h}=O_p(n^{-2|a_h-a_l|}).\nn\eeqn
Thus, by $\sup_{1\leq t\leq n}|\bx_{t1h}|=O_p(n^{a_h-1/2})$ for $h\geq 1$ (see Lemma 7), we have
\beqn \wh\gamma'_i\bA_{1h}\bx_{t1h}=O_p(n^{-a_h+2a_l-1/2})I(h>l)
+O_p(n^{-2a_l+3a_h-1/2})I(1\leq h<l).\nn\eeqn
As a result, by (\ref{7.18}), it follows that for  any $\nu_l+1\leq i\leq \nu_l+r_l$,
\beqn \quad\wh x_{t}^{i}&=& \wh\gamma'_i\bA_{1l}\bx_{t1l}
+O_p(\sum_{h=l+1}^{q}n^{-a_h+2a_l-1/2}+\sum_{h=1}^{l}n^{-2a_{l}+3a_{l-1}-1/2})\nn\\
&=&\wh\gamma'_i\bA_{1l}\bx_{t1l}+O_p(n^{-a_{l+1}+2a_l-1/2}+n^{-2a_{l}+3a_{l-1}-1/2}),\nn\eeqn
where $\bx_{t10}=\bx_{t2}.$ Thus, for  any  given $m$, we have
\beqn \label{7.19} &&\sum_{k=1}^{m}\Big({1\over {n-k}}\sum_{t=1}^{n-k}(\wh x_{t+k}^{ i}-\overline{\wh x}^{i}) (\wh x_{t, i}-\overline{\wh x}^{i})\Big)\nn\\
&=&{\wh\gamma'_i\bA_{1l}\over {n-k}}\sum_{k=1}^{m}\sum_{t=1}^{n-k}(\bx_{t+k, 1l}-\overline{\bx}_{1l})(\bx_{t1l}-\overline{\bx}_{1l})'\bA'_{1l}\wh\gamma_i(1+o_p(1)).\eeqn
By (\ref{7.19}), we have that for any $\nu_l+1\leq i\leq \nu_l+r_l, \, l=1, \cdots, q$
\beqn\label{7.20}&&\sum_{k=1}^{m}\Big({1\over {n-k}}\sum_{t=1}^{n-k}(\wh x_{t+k}^{ i}-\overline{\wh x}^{i}) (\wh x_{t, i}-\overline{\wh x}^{i})\Big)\nn\\
&=&m\wh\gamma'_i\bA_{1l}\Big({1\over n}\sum_{t=1}^{n}(\bx_{t1l}-\overline{\bx}_{1l})(\bx_{t1l}-\overline{\bx}_{1l})'\Big)
\bA'_{1l}\wh\gamma_i(1+o_p(1))
=O_e(mn^{2a_l-1}).\eeqn
On the other hand,  by (\ref{7.19}) and $\|\sum_{k=1}^{m}{1\over {n-k}}\sum_{t=1}^{n-k}(\bx_{t+k, 2}-\overline{\bx}_2)(\bx'_{t, 2}-\overline{\bx}_2)\| \leq C$ in probability, it follows that for $p-r+1\le i\le p$,
\beqn \label{7.21}\sum_{k=1}^{m}\Big({1\over {n-k}}\sum_{t=1}^{n-k}(\wh x_{t+k}^{ i}-\overline{\wh x}^{i}) (\wh x_{t, i}-\overline{\wh x}^{i})\Big)=O_p(1).\eeqn
Equation (\ref{7.20}) together with (\ref{7.21}) yields the conclusion  of Theorem 2 as desired.\hfill$\Box$

\ignore{ Next, we show (b), by the definition of $\wt r$, we have
\beqn \label{6.41add}\sum_{j=1}^{\tilde{r}}\hat{\lambda}_{p+1-j}+(p-\tilde{r}) \omega_n \leq
\sum_{j=1}^{r}\hat{\lambda}_{p+1-j}+(p-r) \omega_n.\eeqn
Suppose that $\tilde{r}<r$, it follows from (\ref{6.41add}) that
\beqn \label{6.42}(r-\tilde{r}) \omega_n\leq \sum_{j=\tilde{r}+1}^{r}\hat{\lambda}_{p+1-j}\leq (r-\tilde{r}) \hat{\lambda}_{p+1-r}.\eeqn
However equation (\ref{6.39}) implies that in probability,
$$\hat{\lambda}_{p+1-r}=\lambda_{p+1-r}(\bA\hat{\bW}^x\bA')=\lambda_{p+1-r}(\hat{\bW}^x)\leq c_4j_0^2.$$
Since $\omega_n/j_0^2\rightarrow \infty$, it follows that equation (\ref{6.42}) holds with probability zero. This gives that
\beqn \label{6.43}\lim_{n\rightarrow\infty}P\{\tilde{r}<r\}=0.\eeqn

On the other hand, if $\tilde{r}>r$, equation (\ref{6.41add}) yields
\beqn \label{6.44}(\tilde{r}-r) \hat{\lambda}_{p-r}\leq\sum_{j=r+1}^{\tilde{r}}\hat{\lambda}_{p+1-j}\leq  (\tilde{r}-r) \omega_n.\eeqn
By (\ref{6.38}), we have when $|I_0|\geq 2$ or $ |I_0|=1$ and
$\mathrm{E}z_t^{I_0}= 0$,
\beqn \label{6.45}\hat{\lambda}_{p-r}=\lambda_{p+1-r}(\hat{\bW}^x)\geq c_3n^{2(2a_1-1)}.\eeqn
A similar argument to (\ref{6.38}) deduces when $ |I_0|=1$ and
$\mathrm{E}z_t^{I_0}\neq 0$,
\beqn \label{6.46}\hat{\lambda}_{p-r}=\lambda_{p+1-r}(\hat{\bW}^x)\geq c'_3n^{4a_1}.\eeqn
Since $\omega_n/n^{2(2d-1)}\rightarrow 0$ as $n\rightarrow\infty$, equations (\ref{6.44})--(\ref{6.46}) imply
\beqn \label{6.47}\lim_{n\rightarrow\infty} P\{\tilde{r}>r\}=0.\eeqn}

\subsection{Proofs for Section \ref{sec32}}

\noindent{\bf Proof of Theorems 3 and 4.}\quad
Theorem 3 can be shown  similarly to Theorem 1 by using Lemma 9   instead of Lemma 8, except that
when $p\rightarrow\infty$,
\beqn ||(\bSigma_0^{x})^{-1} \{\hat{\bSigma}_j^{x}(\hat{\bSigma}_j^{x})'\}^{1/2}||_F
=O_p\left(\left(\sum_{i=1}^{p}(\tilde{\lambda}_i)^2\right)^{1/2}\right)
=O_{p}(p^{1/2}),\nn\eeqn
where $\tilde{\lambda}_i, 1\leq i\leq p$ are solutions of (\ref{6.34}).
As a result, (\ref{6.30a}) should be replaced by
\begin{equation}\label{6.50}||(\bW^{y})^{-1/2}(\hat{\bW})^{1/2}||_F=O_p(p^{1/2}) \, \, \hbox{and} \, \, ||(\bW^{x})^{-1/2}(\hat{\bW}^{x})^{1/2}||_F=O_p(p^{1/2}).\end{equation}

Theorem 4 can be shown similarly to Theorem 2. We omit the details. \hfill$\Box$

\subsection{Proofs for Section \ref{sec5}}

To prove Theorems 5 and 6, we first introduce some notation. Let $ k_{ni}=n^{d_i-1/2}I(d_i>1/2)+n^{d_i+1/2}I(d_i<1/2)$ and $\lambda_i(t- s)=(t-s)^{d_i-1}/\Gamma(d_i)I(d_i>1/2)+(t-s)^{d_i}/\Gamma(d_i+1)I(d_i<1/2)$. Define
$\bK_n=\mathrm{diag}(k_{n1}, \cdots, k_{np}), \, \bLambda(t, s)=\mathrm{diag}(\lambda_1(t-s), \cdots, \lambda_p(t-s))$ and
\beqn \bold{B}_0=0, \, \, \bold{B}_t=(B_t^1, \cdots, B_t^p)'=\int_{0}^{t}\bLambda(t, s)\, d\bold{W}_s, \, \, \bold{U}_t=\bold{B}_t-\int_{0}^{1}\bold{B}_t\, dt,\nn\eeqn
where $\bold{W}_s$ is given in (ii) of Condition 3. Let $\nabla^{d_l}v_{t}^l=\mu_l, \, I_1^{c}=\{i: d_i>1/2\}$ and $\bx_{t, I}=(x_t^i, \, i\in I)'$ and $\bv_{t, I}=(v_t^i, \, i\in I)'$

\vskip3mm

\begin{lemma} Let  $\bZ_n(t)=((\bx_{[nt], I_1^c}-\bv_{[nt], I_1^c})', \sum_{j=1}^{[nt]}(\bx_{j, I_1}-\bv_{j, I_1})')'$. Under (ii) of Condition 3,
\beqn \bK_n^{-1}\bZ_n(t)\stackrel{J_1}{\Longrightarrow} \bold{B}_t, \, \, \hbox{on}\, \, D[0, 1]^{p}.\eeqn
\end{lemma}

\begin{proof} Let $d_{I_1}=\{d_i: \, i\in I_1\}$, then $\sum_{j=1}^{[nt]}\bx_{j, I_1}$ is an integrated fractional process with order $d_{I_1}+1$, and each of its components has order larger than $1/2.$ Using (ii) of Condition 3 instead of  Marinucci and Robinson (2000) Lemma 2, we can show this lemma similarly to  their Theorem 1.
\end{proof}

Let $\bTheta_n$ and  $\bold{M}_i(t)$ be defined as that after Lemma 7 by using
 $H^{d}(t)= {t^{d}/ \Gamma(d_i+1)}- {1/ \Gamma(d_i+2)}$ and
  $\bF_i(t)=(U^{\nu_i+1}(t),  \cdots, U^{\nu_i+r_i}(t))', \,$ where
  $U^{i}_t$ be the $i$-th component of $\bold{U}_t.$ Let $\bL_n=\mathrm{diag}(l_{n1}, \cdots, l_{np}), \, \, l_{ni}=n^{d_i-1/2}I(d_i>1/2)+I(d_i<1/2)$.
  Similar to Lemma 8, by
  Lemma 12 and the continuous mapping theorem, we have the following lemma.

\begin{lemma}
 Let the conditions of Theorem 5 hold. Then the following assertions hold
  for any $0\le j\le j_0$.
\begin{itemize}
\item[(i)] If $\delta>1/2$, then
\beqn &&\bL_n^{-1}\bTheta_n\hat{\bSigma}_j^{x}\bTheta'_n\bL_n^{-1}
\stackrel{d}{\longrightarrow}\int_{0}^{1}(\bold{M}'_t, \bold{U}'_{t2})'(\bold{M}'_t, \bold{U}'_{t2})\, dt,  \quad  \hbox{and} \nn\\
&&\bL_n^{-1}\bTheta_n\bSigma_j^{x}\bTheta'_n\bL_n^{-1}\stackrel{d}{\longrightarrow}\mathrm{diag}\Big(\int_{0}^{1}
\bold{M}_{t}\bold{M}'_{t}\, dt, \, \int_{0}^{1}
\bold{U}_{t, 2}\bold{U}'_{t, 2}\, dt\Big),\nn\eeqn
 where $\bold{U}_{t, 2}$ is corresponding to the  last $p$ components of $\bold{U}_{t}$.

\item[(ii)] If $\delta<1/2$, then
\begin{equation} \label{7.83} \bL_n^{-1}\bTheta_n\hat{\bSigma}_j^{x}\bTheta'_n\bL_n^{-1}\stackrel{d}{\longrightarrow} \mathrm{diag}\left(\int_{0}^{1}\bold{M}_{t}\bold{M}'_{t} \, dt,  \, \, \mathrm{Cov}(\bx_{t+j, I_1}\bx_{t, I_1})\right),\end{equation}
and
\begin{equation}\label{7.84}\bL_n^{-1}\bTheta_n\bSigma_j^{x}\bTheta'_n\bL_n^{-1}\stackrel{d}{\longrightarrow} \mathrm{diag}\left(\int_{0}^{1}\bold{M}_{t}\bold{M}'_{t} \, dt,  \, \, \mathrm{Cov}(\bx_{t+j, I_1}\bx_{t, I_1})\right).\end{equation}
\end{itemize}
\end{lemma}

\vskip3mm
By Lemma 13, Theorems 5 and 6 can  be established
in a similar manner as to Theorems 1 and 2.
Therefore we omit the detailed proofs.

\end{document}